%% file: main.tex
\documentclass[review,onefignum,onetabnum]{siamonline250211}

\input{shared}

\ifpdf
\hypersetup{
  pdftitle={Credible Intervals for Probability of Failure with Gaussian Processes},
  pdfauthor={A. Sorokin, V. Rao}
}
\fi

\nolinenumbers 

\begin{document}

\maketitle

\begin{abstract}
  Estimating the probability of failure for expensive simulations is a central task in reliability analysis for structural design, power grid design, and safety certification, among other areas. This work derives credible intervals on the probability of failure by modeling the simulation as a realization of a Gaussian process surrogate. These intervals are governed by the pointwise binary classification error of the surrogate and are compatible with the broad class of adaptive sampling schemes proposed in the literature. We further propose a novel batch sampling scheme that suggests multiple evaluation points per iteration, enabling parallel simulation on HPC systems. The method is empirically validated using our scalable, open-source implementation on a variety of test problems including a Tsunami model where failure is quantified in terms of maximum wave height.
\end{abstract}

\begin{keywords}
  probability of failure, 
  reliability analysis,
  Gaussian processes, 
  Monte Carlo, 
  probabilistic numerics
\end{keywords}

\begin{MSCcodes}
  65C05, 60G15, 62F15
\end{MSCcodes}

\section{Introduction}
Computing the probability of extreme events is of central importance for complex systems that arise in natural phenomena such as climate, weather, and oceanography \cite{Easterling_2000, Easterling_2000A} as well as in engineering systems such as structures \cite{Cornell_1968, Vrouwenvelder_2000} and power grids \cite{Lesieutre_2008}. Examples of consequential extreme events are rogue waves in the ocean \cite{Dysthe_2008}, hurricanes, tornadoes \cite{Ross_2003}, and power outages \cite{Atputharajah_2009}. Hence, evaluating rare event probabilities is an important task in many scientific fields. Evaluating or simulating the system is often expensive, which necessitates sample efficient estimation schemes. Many methods exist to efficiently find good point estimates. 

Quantifying extreme excursion probabilities is of great importance  because of their  socioeconomic impact. Outcomes of interest reside in the tails of the probability distribution of the associated event space because of their low likelihood. To resolve the tails of these events, one has to evaluate multivariable integrals over complex domains. Because of the tiny mass and complex shape of the relevant likelihood level sets, standard quadrature, cubature, or sparse grid  methods cannot be  applied directly to evaluate these integrals. The most commonly used method is Monte Carlo simulation (MCS), which requires repeated samples of the underlying outcome. For such small probabilities, however, MCS exhibits a large variance relative to the probability to be computed, and thus it needs a large number of samples to produce results of acceptable accuracy. For example, estimating the odds of an extreme event, whose probability ends up being $10^{-3}$ for an underlying process that requires 10 minutes per numerical simulation, requires two years of serial computation for producing an estimate with a standard deviation of less than 10\% of the target value via MCS. Hence, alternative methods must be developed that are computationally efficient. 


The rest of this section reviews related work and outlines our novel contributions. \Cref{sec:Monte_Carlo_Methods} overviews Monte Carlo methods that underpin many of the techniques used for probability of failure estimation. \Cref{sec:Gaussian_Processes} introduces the Gaussian process surrogate model and its connection to pointwise failure classification. Our main contributions start in \Cref{sec:estimators_error_bounds} where we derive credible intervals for common probability of failure estimates. \Cref{sec:adaptive_sampling_algorithm_cost} discusses iterative schemes amenable to the derived credible intervals and proposes a novel batch sampling scheme for simulations on HPC platforms. \Cref{sec:numerical_experiments} empirically evaluates our method on numerous benchmark problems including a Tsunami simulation where probability of failure is quantified in terms of wave height. Finally, \Cref{sec:conclusions_future_work} provides a brief summary alongside proposed directions for future work.

Importance Sampling Monte Carlo (ISMC) is one popular technique for estimating probability of failure; see \cite{Bucklew_2013B,rubinstein2016simulation} for comprehensive treatments of rare event simulation and importance sampling. Basic Monte Carlo techniques often require a large number of samples from the original density which makes them inapplicable to expensive simulations. ISMC can greatly reduce the number of samples required for a good estimate by choosing a sampling density close to the original density truncated to the failure region. The challenge of ISMC schemes is to find a performant sampling density. Popular schemes derive densities from surrogate models (often Gaussian processes) \cite{camera,dalbey2014gaussian,echard2013combined,dubourg2013metamodel}, kernel density estimators \cite{ang1992optimal,swiler2010importance}, importance direction vectors \cite{cheng2023rare}, mixture of Gaussians \cite{wahal2019bimc}, or subset simulation \cite{Au_2007,uribe2021cross} among others. The many techniques of error estimation for Monte Carlo methods can be immediately applied to ISMC, see \cite{qmc_software} for an overview. However, these schemes typically require two separate sets of expensive simulations: one to construct the importance sampling density and another to compute the estimate itself, since reusing the same evaluations for both would introduce bias. 

Another popular approach derives the probability of failure estimate directly from a surrogate model for the simulation. Here simulations are only used to build the surrogate model and the estimator is fully derived from evaluations of the surrogate. Gaussian processes are the most popular choice of surrogate model due to their built-in uncertainty quantification at every point in the input space. A key challenge in this approach is choosing where to evaluate the expensive simulation in order to most efficiently refine the surrogate and the resulting failure probability estimate. This is typically framed as a sequential experimental design problem in which an acquisition function determines the next evaluation location.

A number of acquisition functions have been proposed in this setting. The Expected Feasibility Function (EFF) \cite{bichon2008efficient} targets regions where the Gaussian process posterior mean is close to the failure threshold relative to the posterior standard deviation. Vazquez and Bect \cite{vazquez2009sequential} formulated a sequential Bayesian algorithm specifically for estimating the probability of failure. Stepwise Uncertainty Reduction (SUR) strategies \cite{bect2012sequential} greedily select samples to maximally reduce the posterior uncertainty in the failure probability estimate. Lv et al.\ \cite{lv2015new} proposed an alternative learning function based on the expected risk function for kriging-based reliability analysis. Other contributions include studying the effect of additional samples on uncertainty reduction \cite{bae.pf_gp_uncertainty_reduction} and stochastic spectral embedding for rare event estimation \cite{wagner2022rare}. Polynomial-chaos kriging has also been applied to this setting \cite{schobi2017rare}. See also \cite{camera} for a cost-aware, multi-fidelity extension.

While these methods leverage the Gaussian process to adaptively place samples, they all produce point estimates of the failure probability without rigorous uncertainty quantification on the estimate itself. That is, although the Gaussian process captures pointwise uncertainty about the simulation, none of these methods translate this uncertainty into a credible interval on the resulting probability of failure. To the best of our knowledge, this paper is the first to derive such credible intervals.

Throughout this paper we use the following notations. The set of all natural numbers is denoted by $\bbN$. The transpose of a matrix $\mA$ is $\mA^\top$. For a function $f$, we let $\mathrm{support}(f)$ be the set of all points $\bx$ in the domain for which $f(\bx) \neq 0$. The indicator function $1_F(\bu)$ takes value $1$ if $\bu \in F$ and $0$ otherwise. Bold symbols are used for vectors and $u_j$ is the $j$th element of $\bu$. Inequalities between vectors are taken elementwise, e.g., $\ba \leq \bb$ if and only if $a_j \leq b_j$ for all $j$. Similarly, functions inequalities are taken elementwise, e.g., $f \geq 0$ means $f(\bx) \geq 0$ for all $\bx$ in the domain of $f$. In probability, $\bbE_\bbG$ is the expectation with respect to probability measure $\bbG$. Sometimes we will just write $\bbE$ when the probability measure is understood. $\calU(F)$ is the uniform distribution over the set $F$ while $\calU[0,1]^d$ is the uniform distribution over the $d$ dimensional unit cube. $\calN(\ba,\mSigma)$ is the multivariate Gaussian distribution determined by mean $\ba$ and covariance matrix $\mSigma$ with the dimension implied by the size of the arguments. $\mathrm{Geom}(c)$ is the Geometric distribution with success probability $c$. $\bX_1,\dots,\bX_N \sim Q$ indicates that $\bX_1,\dots,\bX_N$ each have distribution $Q$ and $\bX_1,\dots,\bX_N \simiid Q$ indicates that $\bX_1,\dots,\bX_N$ are independent and identically distributed (IID).

The novel contributions of this work are as follows.
\begin{itemize}
    \item Provide credible intervals that hold with guaranteed confidence for schemes which derive probability of failure estimates from a probabilistic surrogate model.
    \item Provide a novel sampling scheme for iteratively updating a Gaussian process surrogate to approximate probability of failure. This scheme is suitable to high performance computing (HPC) settings where the expensive simulation may be evaluated at multiple parameter configurations in parallel.  
    \item Provide efficient algorithms to update Gaussian process posteriors and their resulting estimates and credible intervals. 
    \item Provide a scalable, open-source implementation. 
\end{itemize}

\section{Monte Carlo Methods} \label{sec:Monte_Carlo_Methods}

Monte Carlo methods are a powerful class of techniques for high dimensional numerical integration or equivalently expectation approximation. Specifically, Monte Carlo methods may efficiently estimate the true mean 
\begin{equation}
    \mu := \bbE\left[f(\bU)\right] = \int_{[0,1]^d} f(\bu) \D \bu
    \label{eq:mc_true_mean}
\end{equation}
where $f: [0,1]^d \to \bbR$ is a measurable integrand and $\bU \sim \calU[0,1]^d$. This setting applies to non-uniform uncertainty since $f$ may incorporate a variable transformation, see \cite{qmc_software} for a discussion and framework. In the probability of failure setting $f = 1_{F}$, meaning $f$ is $1$ on some failure region $F \subseteq [0,1]^d$ and $0$ on the success region $S = [0,1]^d \backslash F$. 

Monte Carlo methods approximate the true mean $\mu$ by the average of function evaluations at specially chosen sampling nodes $\mU_N := \{\bU_i\}_{i=0}^{N-1} \subset [0,1]^d$ to attain the estimator
\begin{equation}
    \widehat{\mu} := \frac{1}{N} \sum_{i=0}^{N-1} f(\bU_i).
    \label{eq:mc_estimate}
\end{equation}
The following subsections discuss the various choices of sampling nodes that flavor different Monte Carlo techniques. More in depth discussions of Monte Carlo methods and the specialized techniques discussed in the following subsections may be found in \cite{niederreiter1992random,dick2010digital,l2014random,l2009monte,hickernell1998generalized,mcbook,qmc_software}.

\subsection{Crude Monte Carlo}

Crude Monte Carlo (CMC) methods choose IID sample nodes $\bU_0,\dots,\bU_{N-1} \simiid \calU[0,1]^d$. We denote \eqref{eq:mc_estimate} subject to this choice by $\widehat{\mu}^\CMC$. This estimator is unbiased for $\mu$ and has variance $\Var[f(\bU)]/N$; in the probability of failure setting $\Var[\widehat{\mu}^\CMC] = \mu(1-\mu)/N$. When $\Var[f(\bU)] < \infty$, the approximation error $\lvert \mu - \widehat{\mu}^\CMC \rvert$ converges to $0$ at the rate $\calO(N^{-1/2})$. For example, consider a probability of failure $\mu = 10^{-3}$ for a simulation requiring $10$ minutes per evaluation. Achieving a standard deviation of $10^{-4}$ (i.e., $10\%$ relative error) via CMC requires $N = \mu(1-\mu)/10^{-8} \approx 10^{5}$ samples, corresponding to roughly two years of serial computation. This renders CMC impractical for rare event probability estimation with expensive forward models.

\subsection{Quasi-Monte Carlo} \label{sec:QMC}

Quasi-Monte Carlo (QMC) methods carefully coordinate dependent sample nodes $\mU_N = \{\bU_i\}_{i=0}^{N-1}$ so their discrete distribution $\frac{1}{N} \sum_{i=0}^{N-1} 1_{\{\bv \in [0,1]^d: \bv \leq \bU_i\}}(\bu)$ is ``closer to'' the true distribution $\prod_{j=1}^d u_j$ across all $\bu \in [0,1]^d$. The distance between the discrete and true distributions is quantified by a discrepancy measure. The Koksma--Hlawka inequality \cite{hlawka1961funktionen} says that for any $f$ with bounded variation $V_\mathrm{HK}(f)$ in the sense of Hardy and Krause, we have 
$$\left\lvert \widehat{\mu}^\QMC - \mu \right\rvert \leq V_\mathrm{HK}(f) D^*(\mU_N)$$
where 
$$D^*(\mU_N) = \sup_{\bu \in [0,1]^d} \left\lvert \prod_{j=1}^d u_j - \frac{1}{N} \sum_{i=0}^{N-1} 1_{\{\bv \in [0,1]^d: \bv \leq \bU_i\}}(\bu) \right\rvert$$
is the star-discrepancy and $\widehat{\mu}^\QMC$ is \eqref{eq:mc_estimate} under the current choice of QMC nodes. Other discrepancy, variation pairings are also available, see \cite{dick2010digital,hickernell1998generalized} for an overview. Extensible constructions exist to attain 
\begin{equation}
    D^*(\mU_N) = \calO((\log N)^d / N),
    \label{eq:Dstar_rate}
\end{equation}
so for an integrand $f$ with finite variation $V_\mathrm{HK}(f)$ we achieve $\lvert \mu - \widehat{\mu}^\QMC \rvert = \calO(N^{-1+\delta})$ for any $\delta > 0$. This rate is significantly faster than that of Crude Monte Carlo for this more restrictive class of nicely behaved integrands. 

This work utilizes sampling nodes from an extensible, randomized low discrepancy (LD) sequence which attains the desired discrepancy rate in \eqref{eq:Dstar_rate}. The method of randomization ensures $\bU_0,\dots,\bU_{N-1} \sim \calU[0,1]^d$ so $\widehat{\mu}^\QMC$ remains unbiased for $\mu$. Randomization deceases the likelihood of the points badly matching the integrand and ensures $\mU_N \subset (0,1)^d$ with probability $1$. These LD point sets are \emph{not} significantly slower to generate than IID points. Popular LD sequences include the Sobol' (a special case of a digital net), lattice, and Halton constructions. LD sequences and some of their randomization techniques are given a more thorough treatment in \cite{niederreiter1992random,dick2010digital,l2014random,sorokin.2025.ld_randomizations_ho_nets_fast_kernel_mats}.

\subsection{Importance Sampling Monte Carlo}

Importance Sampling Monte Carlo (ISMC) replaces $\mU_N$ in \eqref{eq:mc_estimate} with sampling nodes $\bX_0,\dots,\bX_{N-1} \simiid Q$ and adjusts the weights in the sum to remain unbiased. Here $Q$ is some auxiliary distribution with density $q$ where we assume $\mathrm{support}(f) \subseteq \mathrm{support}(q)$. We may rewrite the true mean in \eqref{eq:mc_true_mean} as 
$$\mu = \int_{[0,1]^d} f(\bu) \D \bu =  \int_{[0,1]^d} \frac{f(\bx)}{q(\bx)} q(\bx) \D \bx = \bbE\left[\frac{f(\bX)}{q(\bX)}\right]$$
where $\bX \sim Q$. Plugging in nodes $\mX:=\{\bX_i\}_{i=0}^{N-1}$ and adjusting evaluation weights produces the unbiased estimator 
$$\widehat{\mu}^\ISMC := \frac{1}{N} \sum_{i=0}^{N-1} \frac{f(\bX_i)}{q(\bX_i)}$$
with
$$\Var\left[\widehat{\mu}^\ISMC\right] = \frac{1}{N} \Var\left[\frac{f(\bX)}{q(\bX)}\right] = \frac{1}{N}\bbE\left[\left(\frac{f(\bX)}{q(\bX)}-\mu\right)^2\right] = \frac{1}{N} \int_{[0,1]^d} \frac{\left[f(\bx)-\mu q(\bx)\right]^2}{q(\bx)}\D \bx.$$
When $\Var[f(\bX)/q(\bX)]$ is less than $\Var[f(\bU)]$, ISMC is expected to outperform CMC for fixed $N$. However, a bad choice of $Q$ may result in a large variance and render ISMC far less efficient than CMC. 

In the probability of failure setting with $f=1_F \geq 0$, $\Var[\widehat{\mu}^\ISMC]=0$ when $q = 1_{F} / \mu$. Since $F$ and $\mu$ are generally unknown, the optimal importance sampling distribution $\calU(F)$ is intractable but may still provide guidance to methods which utilize ISMC for probability of failure estimation. 

As discussed in the introduction, many methods in the probability of failure literature deploy ISMC as a means to reduce the number of expensive simulations required for a good approximation \cite{dubourg2013metamodel,uribe2021cross,camera,ang1992optimal,cheng2023rare,dalbey2014gaussian,echard2013combined,swiler2010importance}. Often, such methods will use an initial set of simulation evaluations to propose an importance sampling density. A separate set of simulations must then be run to obtain the ISMC estimator. Such methods may be made adaptive by constructing separate ISMC estimates at after iteratively increasing the sampling size. However, the optimal linear combination of such ISMC estimates simply selects the lowest variance estimate, thus disregarding both evaluations used to fit the probabilistic model and evaluations from other ISMC estimates throughout the adaptive routine. 
 
\subsection{Rejection Sampling}\label{sec:rs}

Rejection sampling is a technique for drawing IID samples from $Q$ whose density $q$ is only known up to a constant multiple. \Cref{alg:RS} details a rejection sampling procedure for drawing IID nodes from an unnormalized density $\varrho \leq 1$ for $Q$ \cite{practicalqmc}.
\begin{algorithm}[htbp]
    \caption{$\boldsymbol{\mathrm{AlgRS}}(\varrho,b)$: Rejection Sampling} \label{alg:RS}
    \begin{algorithmic}[1]
        \REQUIRE{$\varrho:[0,1]^d \to [0,1]$ \COMMENT{unnormalized density to sample from}}
        \REQUIRE{$b \in \bbN$ \COMMENT{the number of samples to draw}}
        \STATE{$\mX_0 \gets \emptyset$ \COMMENT{initialize an empty point set of IID draws from $\varrho$}}
        \STATE{$i \gets 0$ \COMMENT{initialize accepted counter}}
        \STATE{$t \gets 0$ \COMMENT{initialize tried counter}}
        \WHILE{$i<b$}
            \STATE{draw $\bX_t \sim \calU[0,1]^d$ independent of $\{\bX_j: j<t\}, \{\bU_j: j<t\}$ \COMMENT{draw candidate}}
            \STATE{draw $\bU_t \sim \calU(0,1)$ independent of $\{\bX_j: j \leq t\}, \{\bU_j: j < t\}$ \COMMENT{draw threshold}}
            \IF{$\bU_t \leq \varrho(\bX_t)$}
              \STATE{$i \gets i+1$ \COMMENT{increment accepted counter}}
              \STATE{$\mX_i \gets \mX_{i-1} \cup \{\bX_t\}$ \COMMENT{add candidate to the accepted set}}
            \ENDIF
            \STATE{$t \gets t+1$ \COMMENT{increment tried counter}}
        \ENDWHILE
        \RETURN{$\mX_{b}$ \COMMENT{return the $b$ IID draws from $\varrho$}}
    \end{algorithmic}
\end{algorithm}

The random number of tries $T$ required to draw a single sample from $Q$ with such rejection sampling follows $T \sim \mathrm{Geom}(c)$ where $c = \bbE\left[\varrho(\bU)\right]$ is the rejection sampling efficiency with $\bU \sim \calU[0,1]^d$. Then one expects that drawing $b$ IID points from $\varrho$ will require $b/c$ evaluations of $\varrho$.

\section{Gaussian Process Surrogate Model} \label{sec:Gaussian_Processes}

\subsection{Gaussian Process Regression} \label{sec:gp_regression}
\begin{figure}[htbp]
    \centering
    \includegraphics[width=\textwidth]{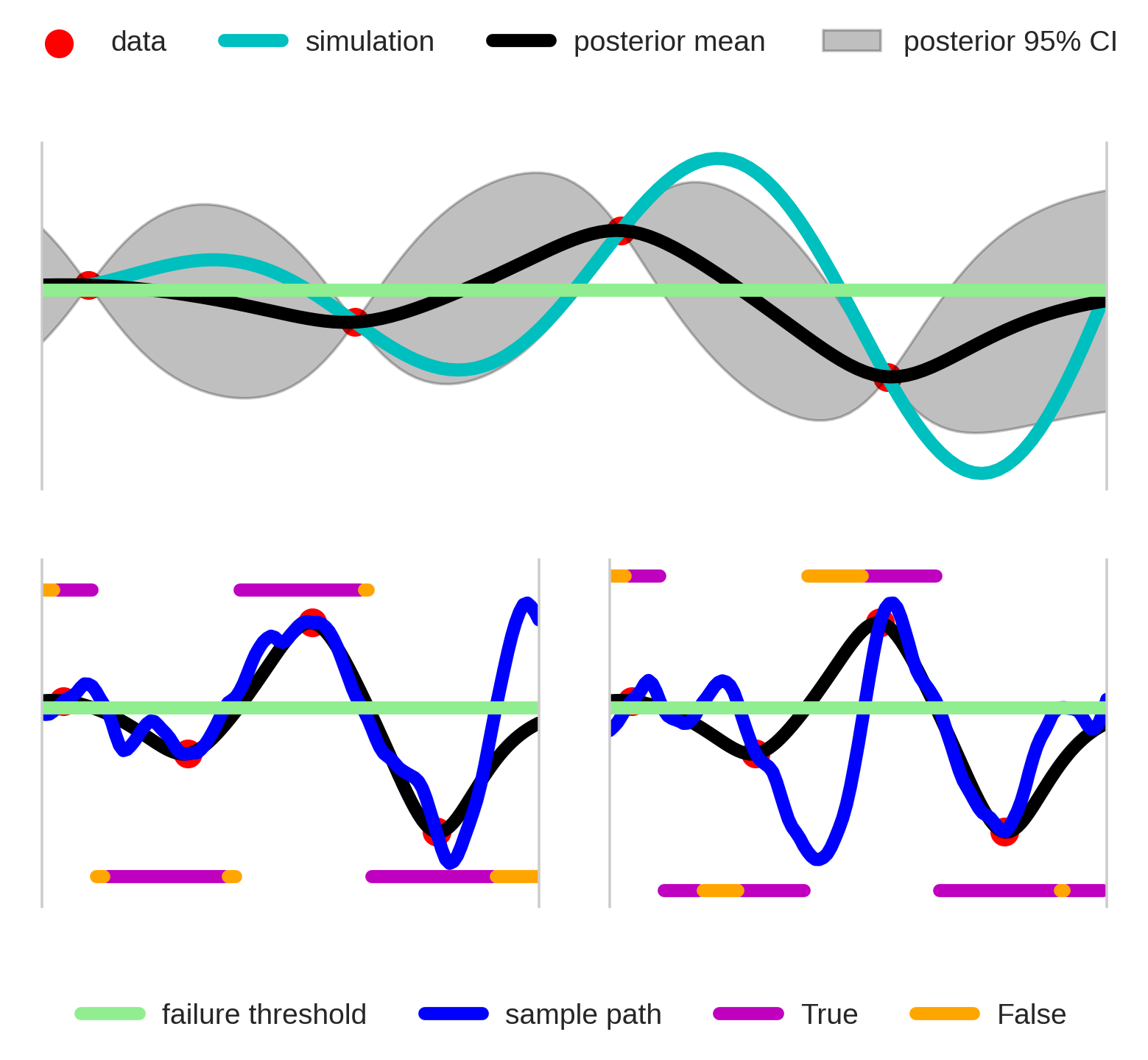}
    \caption{The top figure shows the true simulation and a posterior Gaussian process fit to a few data points. The bottom two plots show some Gaussian process sample paths $g_1,g_2 \in \Omega$ with their corresponding $\TP$, $\FP$, $\TN$, and $\FN$ regions. The predicted failure $\TP$ and $\FP$ regions are shown along the top of each plot while the predicted success $\TN$ and $\FN$ regions are shown along the bottom of each plot.}
    \label{fig:TP_FP_TN_FN}
\end{figure}

Gaussian process regression, or kriging, is a Bayesian technique for modeling both the knowledge of a function at given locations and the uncertainty in a function between those locations. Gaussian processes in machine learning are given a thorough treatment in \cite{rasmussen.gp_ml}, from which the majority of this section is derived. 

Let $g:[0,1]^d \to \bbR$ be a canonical real valued stochastic Gaussian Process on probability space $(\Omega,\calF,\bbG_0)$ indexed by elements in $[0,1]^d$. The Gaussian process regression model is specified by a \emph{prior mean} $m_0: [0,1]^d \to \bbR$ and positive definite \emph{prior covariance kernel} $k_0:[0,1]^d \times [0,1]^d \to \bbR$ so that for any deterministic $\bu,\bu_1,\bu_2 \in [0,1]^d$ we have
\begin{align*}
    m_0(\bu) &:= \bbE_{\bbG_0}\left[g(\bu) \right], \\
    k_0(\bu_1,\bu_2) &:= \Cov_{\bbG_0}\left[g(\bu_1), g(\bu_2) \right], \quad\mathrm{and} \\
    \sigma^2_0(\bu) &:= \Var_{\bbG_0}\left[g(\bu) \right] = k_0(\bu,\bu).
\end{align*}

Now assume we observed $g$ at $\mX := \mX_n = \{\bX_i\}_{i=0}^{n-1} \in [0,1]^{n \times d}$ to get $\by := \by_n  = \{y_i\}_{i=0}^{n-1} \in \bbR^n$ where $y_i = g(\bX_i)$. Denote by $\bbG_n$ the conditional distribution of $g$ given data $\{\mX,\by\}$ under which $g$ remains a Gaussian process. Let us use the following notations:
\begin{itemize}
    \item $\boldm_{\mX} \in \bbR^{n \times 1}$ is the prior mean vector with row $i$ equal to $m_0(X_i)$.
    \item $\mK_{\mX,\mX} \in \bbR^{n \times n}$ be the prior kernel matrix with the element in row $i$, column $j$ equal to $k_0(\bX_i,\bX_j)$.
    \item $\bk_{\mX}(\bu) \in \bbR^{n \times 1}$ is the kernel vector with row $i$ equal to $k_0(\bX_i,\bu)$ for $\bu \in [0,1]^d$.
\end{itemize}
Then for any $\bu,\bu_1,\bu_2 \in [0,1]^d$ the \emph{posterior mean, covariance, and variance} are   
\begin{align*}
    m_n(\bu) &:= \bbE_{\bbG_n}\left[g(\bu)\right] = \bk_{\mX}^\top (\bu)\mK_{\mX,\mX}^{-1}(\by-\boldm_{\mX})+m_0(\bu), \\
    k_n(\bu_1,\bu_2) &:= \Cov_{\bbG_n}\left[g(\bu_1),g(\bu_2)\right] = k_0(\bu_1,\bu_2)-\bk_{\mX}(\bu_1)^\top \mK_{\mX,\mX}^{-1} \bk_{\mX}(\bu_2), \quad\mathrm{and} \\
    \sigma_n^2(\bu) &:= \Var_{\bbG_n}\left[g(\bu) \right] = k_n(\bu,\bu).
\end{align*}
Notice that the posterior covariance and variance only depend on the observation locations $\mX$ and not the observations $\by$. The top panel of \cref{fig:TP_FP_TN_FN} visualizes a Gaussian process with posterior mean $m_n(\bu)$ and a $95\%$ confidence interval $[m_n(\bu) - 1.96\sigma_n(\bu),m_n(\bu)+1.96\sigma_n(\bu)]$ around every point $\bu \in [0,1]^d$. 

\subsection{Binary Classification with Gaussian Processes} \label{sec:binary_classificaiton_GP}

With the Gaussian process posterior in hand, we now formalize how it induces a pointwise classification of the input space into predicted failure and success regions. This classification leads to natural notions of accuracy and error rate that play a central role in the credible intervals of \cref{sec:estimators_error_bounds} and the adaptive sampling scheme of \cref{sec:adaptive_sampling_algorithm_cost}.

Two sources of uncertainty are present. In the Monte Carlo setting of \cref{sec:Monte_Carlo_Methods} there is uncertainty in the canonical $\bU$ on probability space $([0,1]^d,\calB([0,1]^d),\bbU)$ where $\bbU$ is the Lebesgue measure on $[0,1]^d$ so $\bU \sim \calU[0,1]^d$. In \cref{sec:gp_regression} there is uncertainty in the canonical stochastic Gaussian process $g$ on probability space $(\Omega,\calF,\bbG)$ indexed by elements in $[0,1]^d$ with mean $m: [0,1]^d \to \bbR$ and positive definite covariance kernel $k: [0,1]^d \times [0,1]^d \to \bbR$. We refer to $\bbU$ as the domain probability over the parameter space and $\bbG$ as the range probability over the function space.
Let us define the random \emph{failure and success domains} $F,S: \Omega \to \calB([0,1]^d)$ by
\begin{align*}
    F(g) &:= \{\bu \in [0,1]^d: g(\bu) \geq 0\}, \\
    S(g) &:= \{\bu \in [0,1]^d: g(\bu) < 0\}
\end{align*}
respectively. Then the \emph{domain probability of failure} $P: \Omega \to [0,1]$ is 
\begin{equation}
    P(g) := \bbU(F(g)) = \bbE_\bbU[1_{F(g)}(\bU)].
    \label{eq:P(g)}
\end{equation}
and the \emph{range probability of failure} $p: [0,1]^d \to [0,1]$ is 
\begin{equation}
    p(\bU) := \bbE_\bbG[1_{F(g)}(\bU)] = \Phi\left(\frac{m(\bU)}{\sigma(\bU)}\right)
    \label{eq:p(U)}
\end{equation}
where $\Phi$ is the CDF of a standard normal distribution and $\sigma^2(\bU) := k(\bU,\bU)$ as in \cref{sec:Gaussian_Processes}. The framework presented here applies to more general probabilistic models when the last equality above is replaced appropriately for non-Gaussian processes. 

Let the deterministic \emph{predicted domain failure, success regions} $\widehat{F},\widehat{S} \subseteq [0,1]^d$ be
\begin{align*}
    \widehat{F} &:= \{\bu \in [0,1]^d: p(\bu) \geq 1/2 \} = \{\bu \in [0,1]^d: m(\bu) \geq 0\}, \\
    \widehat{S} &:= \{\bu \in [0,1]^d: p(\bu) < 1/2\} = \{\bu \in [0,1]^d: m(\bu) < 0\}.
\end{align*}
Then we may define the \emph{true positive, true negative, false positive, and false negative} functions $\TP,\FP,\TN,\FN: \Omega \to \calB([0,1]^d)$ by  
\begin{alignat*}{2}
    \TP(g) &= \widehat{F} \cap F(g), \qquad 
    \FP(g) &= \widehat{F} \cap S(g), \\
    \TN(g) &= \widehat{S} \cap S(g), \qquad 
    \FN(g) &= \widehat{S} \cap F(g).
\end{alignat*}
Here positive indicates failure and negative indicates success. The bottom two panels of \cref{fig:TP_FP_TN_FN} visualize these four disjoint regions covering $[0,1]^d$ for some $g_1,g_2 \in \Omega$. 

Notice that $1_{\TP(g)}(\bU) + 1_{\TN(g)}(\bU) + 1_{\FP(g)}(\bU) + 1_{\FN(g)}(\bU) = 1$. Then the \emph{expected accuracy} $\ACC: [0,1]^d \to [0,1]$ is naturally defined as
\begin{align*}
  \ACC(\bU) &:= \bbE_\bbG \left[1_{\TP(g)}(\bU)+1_{\TN(g)}(\bU)\right] \\
  &= 1_{\widehat{F}}(\bU)p(\bU) + 1_{\widehat{S}}(\bU)[1-p(\bU)] \\
  &= \max\{p(\bU),1-p(\bU)\}
\end{align*}
and the \emph{expected error rate} $\ERR: [0,1]^d \to [0,1]$ is defined as 
\begin{equation}
    \ERR(\bU) := 1-\ACC(\bU) = 1_{\widehat{F}}(\bU)[1-p(\bU)] + 1_{\widehat{S}}(\bU)p(\bU) = \min\{p(\bU),1-p(\bU)\}.
    \label{eq:ERR}
\end{equation}
The $\max$ and $\min$ expressions result from the fact that $\bU \in \widehat{F}$ if and only if $p(\bU) \geq 1/2$.

\section{Estimators and Credible Intervals} \label{sec:estimators_error_bounds}

Our main contributions are \cref{thm:ci_P_check} and \cref{thm:ci_P_hat} in this section which derive error bounds for the domain probability of failure quantity of interest
$$P(g) = \bbE_\bbU[1_{F(g)}(\bU)]$$ 
as defined in \eqref{eq:P(g)}. Specifically, we derive credible intervals of the form $[\underline{P},\overline{P}]$ so that 
\begin{equation}
    \bbG\left(P \in [\underline{P},\overline{P}]\right) \geq 1-\alpha.
    \label{eq:ci}
\end{equation}
where $\alpha \in (0,1)$ is some uncertainty threshold. Similar variance bounds were derived in \cite{bect2012sequential} where they were used in a greedy sequential sampling scheme. The bounds derived in \cref{thm:ci_P_hat} will always be tighter than those in \cref{thm:ci_P_check}. 

\begin{theorem}[Credible Interval from Posterior Mean Estimate] \label{thm:ci_P_check}
    Denote the posterior mean estimate by
    \begin{equation}
        \check{P} := \bbE_\bbG[P(g)] = \bbE_\bbG[\bbE_\bbU[1_{F(g)}(\bU)]] = \bbE_\bbU[p(\bU)].
        \label{eq:P_check}
    \end{equation} 
    Then \eqref{eq:ci} holds when 
    \begin{equation}
        \underline{P} = \max\left\{\check{P}-\check{\gamma},0\right\}, \qquad \overline{P} = \min\left\{\check{P}+\check{\gamma},1\right\}, \qquad \check{\gamma} := \frac{2\bbE_\bbU\left[p(\bU)(1-p(\bU))\right]}{\alpha}. 
        \label{eq:p_bounds_check}
    \end{equation}
\end{theorem}
\begin{proof}
    Markov's inequality, Jensen's inequality, and the triangle inequality imply that for $\gamma > 0$ we have 
    \begin{align*}
        \gamma \bbG\left(\left\lvert P -\check{P} \right\rvert \geq \gamma \right) &\leq \bbE_\bbG \left[\left\lvert P -\check{P} \right\rvert\right] \\
        &= \bbE_\bbG\left[\left\lvert \bbE_\bbU [1_{F(g)}(\bU)] - \bbE_\bbU [p(\bU)] \right\rvert \right] \\
        &= \bbE_\bbG\left[\left\lvert \bbE_\bbU \left[1_{F(g)}(\bU) - p(\bU)\right]\right\rvert\right] \\
        &= \bbE_\bbG\left[\left\lvert \bbE_\bbU\left[1_{F(g)}(\bU)\left(1-p(\bU)\right) - 1_{S(g)}(\bU)p(\bU)\right] \right\rvert \right] \\
        &\leq \bbE_\bbG\left[\bbE_\bbU\left[1_{F(g)}(\bU)\left(1-p(\bU)\right) + 1_{S(g)}(\bU)p(\bU)\right] \right] \\
        &= \bbE_\bbU\left[\bbE_\bbG\left[1_{F(g)}(\bU)\right]\left(1-p(\bU)\right) + \bbE_\bbG\left[1_{S(g)}(\bU)\right]p(\bU)\right] \\
        &= 2\bbE_\bbU\left[p(\bU)\left(1-p(\bU)\right)\right].
    \end{align*}
\end{proof}

\begin{theorem}[Credible Interval from Predicted Failure Region Estimate] \label{thm:ci_P_hat}
    Denote the probability of the predicted failure region estimate by 
    \begin{equation}
        \widehat{P} := \bbU(\widehat{F}) = \bbE_\bbU[1_{\widehat{F}}(\bU)]. \label{eq:P_hat}
    \end{equation}
    Then \eqref{eq:ci} holds when 
    \begin{equation}
        \underline{P} = \max\left\{\widehat{P}-\widehat{\gamma},0\right\}, \qquad \overline{P} = \min\left\{\widehat{P}+\widehat{\gamma},1\right\}, \qquad \widehat{\gamma} := \frac{\bbE_\bbU\left[\ERR(\bU)\right]}{\alpha}
        \label{eq:p_bounds_hat}
    \end{equation}
    for $\ERR(\bU) = \min\{p(\bU),1-p(\bU)\}$ as defined in \eqref{eq:ERR}.
\end{theorem}
\begin{proof}
    Markov's inequality, Jensen's inequality, and the triangle inequality imply that for $\gamma > 0$ we have 
    \begin{align*}
        \gamma \bbG\left(\left\lvert P - \widehat{P} \right\rvert \geq \gamma \right) &\leq \bbE_\bbG \left[\left\lvert P - \widehat{P} \right\rvert\right] \\
        &= \bbE_\bbG\left[\left\lvert \bbE_\bbU\left[1_{F(g)}(\bU) - 1_{\widehat{F}}(\bU)\right]\right\rvert\right] \\
        &= \bbE_\bbG\left[\left\lvert \bbE_\bbU\left[1_{F(g)}(\bU)1_{\widehat{S}}(\bU) - 1_{S(g)}(\bU)1_{\widehat{F}}(\bU)\right]\right\rvert\right] \\
        &\leq \bbE_\bbG\left[\bbE_\bbU\left[1_{F(g)}(\bU)1_{\widehat{S}}(\bU) + 1_{S(g)}(\bU)1_{\widehat{F}}(\bU)\right]\right] \\
        &= \bbE_\bbU\left[\bbE_\bbG\left[1_{F(g)}(\bU)1_{\widehat{S}}(\bU) + 1_{S(g)}(\bU)1_{\widehat{F}}(\bU)\right]\right] \\
        &= \bbE_\bbU[\bbE_\bbG[1_\FN(\bU)+1_\FP(\bU)]] \\
        &= \bbE_\bbU[\ERR(\bU)].
    \end{align*}
\end{proof}

Since for any $\bU \in [0,1]^d$ either $p(\bU) \geq 1/2$ or $1-p(\bU) \geq 1/2$, we see that $\min\{p(\bU),1-p(\bU)\} \leq 2p(\bU)\left(1-p(\bU)\right)$. Therefore, the credible interval in \cref{thm:ci_P_hat} is tighter than the credible interval in \cref{thm:ci_P_check}, and going forward we take $[\underline{P},\overline{P}]$ to be defined as in \eqref{eq:p_bounds_hat}. 

\section{Adaptive Algorithm} \label{sec:adaptive_sampling_algorithm_cost}

The backbone of our method sequentially updates a surrogate Gaussian process to refine the probability of failure estimate and shrink the resulting credible interval. The general iterative procedure is as follows 
\begin{description}
    \item[\textbf{Input}] a Gaussian process prior specified by a prior mean function $m_0$ and prior covariance function $k_0$ which determine prior distribution $\bbG_0$.
    \item[\textbf{Input}] a simulation $g$ which we assume is a realization of the Gaussian process. The prior number of samples of $g$ is $n \gets 0$. 
    \item[\textbf{Step 1}] Evaluate $g$ at some batch of nodes $\bX_1,\dots,\bX_b$ and set $n \to n+b$.
    \item[\textbf{Step 2}] Update the Gaussian process distribution to $\bbG_n$ based on all previous evaluations of $g$.
    \item[\textbf{Step 3}] Compute $N$ sample QMC approximates $\widehat{P}^\QMC_n,\widehat{\gamma}^\QMC_{n}$ of $\widehat{P}_n,\widehat{\gamma}_{n}$ as defined in \eqref{eq:P_hat}, \eqref{eq:p_bounds_hat} based on posterior distribution $\bbG_n$. 
    \item[\textbf{Step 4}] If the approximate $1-\alpha$ credible interval 
    $$[\underline{P}^\QMC_{n},\overline{P}^\QMC_{n}] := [\max\{\widehat{P}^\QMC_n-\widehat{\gamma}^\QMC_{n},0\},\min\{\widehat{P}^\QMC_n+\widehat{\gamma}^\QMC_{n},1\}]$$
    is desirably narrow or the budget for sampling $g$ has expired, we are done. Otherwise, return to Step 1 to continue refining the approximate estimate and credible interval. 
\end{description}

The choice of sampling scheme in Step 1 is arbitrary. A number of deterministic one step, i.e., $b=1$, look ahead schemes are proposed in \cite{bae.pf_gp_uncertainty_reduction,bect2012sequential,wagner2022rare,vazquez2009sequential,bichon2008efficient,lv2015new,camera,dalbey2014gaussian,echard2013combined}. While some are theoretically applicable to $b>1$, the optimization required in practice quickly becomes intractable. We wish to allow $b>1$ so parallel implementations of simulation $g$ on HPC systems may be fully utilized. To this end, we propose to sample $b$ IID points from the distribution with unnormalized density $2\ERR_n \leq 1$ as defined in \eqref{eq:ERR} under $\bbG_n$. IID samples may be obtained using the rejection sampling techniques of \cref{sec:rs}. The expected number of tries required to draw $b$ IID points from $\ERR_n$ is $b/(2\bbE_\bbU[\ERR_n(\bU)])=b/(2\alpha \widehat{\gamma}_n)$. This randomized, non-greedy scheme shows strong empirical performance in the next section. 

Despite not being written explicitly, the estimates $\widehat{P}_n^\QMC$ and $\widehat{\gamma}_n^\QMC$ in Step 3 depend on the desired uncertainty threshold $\alpha$, the posterior distribution $\bbG_n$, and the $N$ QMC sampling nodes $\mU_N := \{\bU_i\}_{i=0}^{N-1}$. We choose to leave the low discrepancy nodes $\mU_N$ unchanged across iterations. $N$ is chosen to be large, e.g., $N=2^{20}$, as the QMC estimates are evaluated using only the posterior mean $m_n$ and posterior covariance $k_n$. 

The prior mean $m_0$ and prior kernel $k_0$ often involve hyperparameters $\eta$ which may be optimized to better fit the data. For example, in the next section we set a prior mean of $0$ and use the Mat\'ern kernel
$$k_0(\bu_1,\bu_2) = \frac{1}{\Gamma(\nu)}2^{\nu-1} \left(\frac{\sqrt{2\nu}}{l}\lVert \bu_1-\bu_2 \rVert_2\right)^\nu K_\nu\left(\frac{\sqrt{2\nu}}{l}\lVert \bu_1-\bu_2 \rVert_2\right), \qquad \bu_1,\bu_2 \in [0,1]^d$$
with hyperparameters $\eta = (\nu,l)$ where $K_\nu$ is the modified Bessel function and $\Gamma$ is the gamma function. 

One may choose to re-optimize hyperparameters at each iteration in Step 2. This costs $\calO(n^3)$ as the kernel matrix decomposition must be completely recomputed. Alternatively, the hyperparameters may be optimized after the initial sampling and kept fixed in subsequent iterations. This costs $\calO(b^3+(n-b)^2b)$ since we may reuse the kernel matrix decomposition from the previous $n-b$ samples. Moreover, since the QMC nodes $\mU_N$ are unchanged, we may efficiently update the QMC approximates by using $\{p_{n-b}(\bU_i)\}_{i=0}^{N-1}$ when computing $\{p_n(\bU_i)\}_{i=0}^{N-1}$. 

\begin{remark}
If the sampling nodes in Step $1$ are instead chosen to match the kernel, e.g., from a low discrepancy sequence with a matching isotropic kernel, then the kernel matrix decomposition costs only $\calO(n \log n)$ \cite{rathinavel2019fast,jagadeeswaran2019fast,jagadeeswaran2022fast}. However, such a sampling scheme is not posterior-aware and will typically require a much larger number of expensive simulations than adaptive schemes that trade off exploration and exploitation.
\end{remark}

\section{Numerical Experiments}\label{sec:numerical_experiments}

The following numerical experiments are based on our open source Python implementation in the QMCPy Python package \cite{QMCPy.software}. Benchmark examples are reproducible in \url{https://github.com/QMCSoftware/QMCSoftware/blob/master/demos/talk_paper_demos/ProbFailureSorokinRao/prob_failure_gp_ci.ipynb}.  Since $\ERR_0 \propto 1_{[0,1]^d}$, a uniform density, the initial samples are selected from a randomized low discrepancy sequence rather than drawing IID points using rejection sampling as indicated in the algorithm. 

In the implementation, Gaussian process regression is performed using GPyTorch \cite{gardner2018gpytorch}. GPyTorch enables flexible Gaussian process construction with a variety of available prior mean functions, prior covariance kernels, and optimization techniques. Moreover, GPyTorch seamlessly enables GPU utilization making our algorithm scalable to thousands of samples. 

A number of benchmarks in small and medium dimensions are listed in \cref{table:toy_examples} alongside their true solution and corresponding figure reference. Such examples are often defined by original function $\tilde{g}: \calT \to \bbR$ with respect to non-uniform random variable $\bV$ on $\calT$ and where failure occurs when the simulation exceeds some $\eta \in \bbR$. We may perform a change of variables so $\bV \sim \varphi(\bU)$ where $\bU \sim \calU[0,1]^d$ as desired and the probability of failure of $g(\bU) := \tilde{g}(\varphi(\bU))-\eta$ is equivalent to the probability of failure of $\tilde{g}(\bV)$. For example, if $\bV \sim \calN(\ba,\mSigma)$ then one may choose $\varphi(\bU) = \ba + \mA \Phi^{-1}(\bU)$ where $\mSigma = \mA\mA^\top$ and $\Phi^{-1}$ is the inverse CDF of the standard normal taken elementwise. See \cite{qmc_software} for a more rigorous framework and further examples.

\begin{table}[t]
    \centering
    \begin{tabular}{c c c c}
        problem & dimension & true $P(g)$ & Figure\\
        \hline
        Sine & $1$ & $0.50$ & \Cref{fig:sine} \\
        Multimodal \cite{bichon2008efficient} & $2$ & $0.30$ & \Cref{fig:multimodal} \\
        Four Branch \cite{schobi2017rare} & $2$ & $0.21$ & \Cref{fig:fourbranch} \\
        Ishigami \cite{ishigami1990importance} & $3$ & $0.16$ & \Cref{fig:ishigami} \\
        Hartmann \cite{balandat2020botorch} & $6$ & $0.0074$  & \Cref{fig:hartmann} \\
        \hline
    \end{tabular}
    \caption{Benchmark problems across a variety of dimensions.}
    \label{table:toy_examples}
\end{table}

As a more realistic example, we look at the probability the maximum height of a Tsunami exceeds $3$ meters SSHA (sea surface height anomaly) subject to uniform uncertainty on the origin of the Tsunami. The Tsunami simulation was run using UM-Bridge \cite{Seelinger2023}, an interface deploying containerized models. The Tsunami simulation is detailed  in \cite{seelinger2021high}. 

In all the examples we see the credible interval captures the true mean and the width of the credible interval decreases steadily across iterations. The true error is usually at least an order of magnitude below the error indicated from the credible interval. We caution that these results are highly dependent on the choice of the prior kernel, prior mean and their hyperparameters. 

\begin{figure}[H]
    \centering
    \textbf{Sine Problem}
    \includegraphics[width=\textwidth]{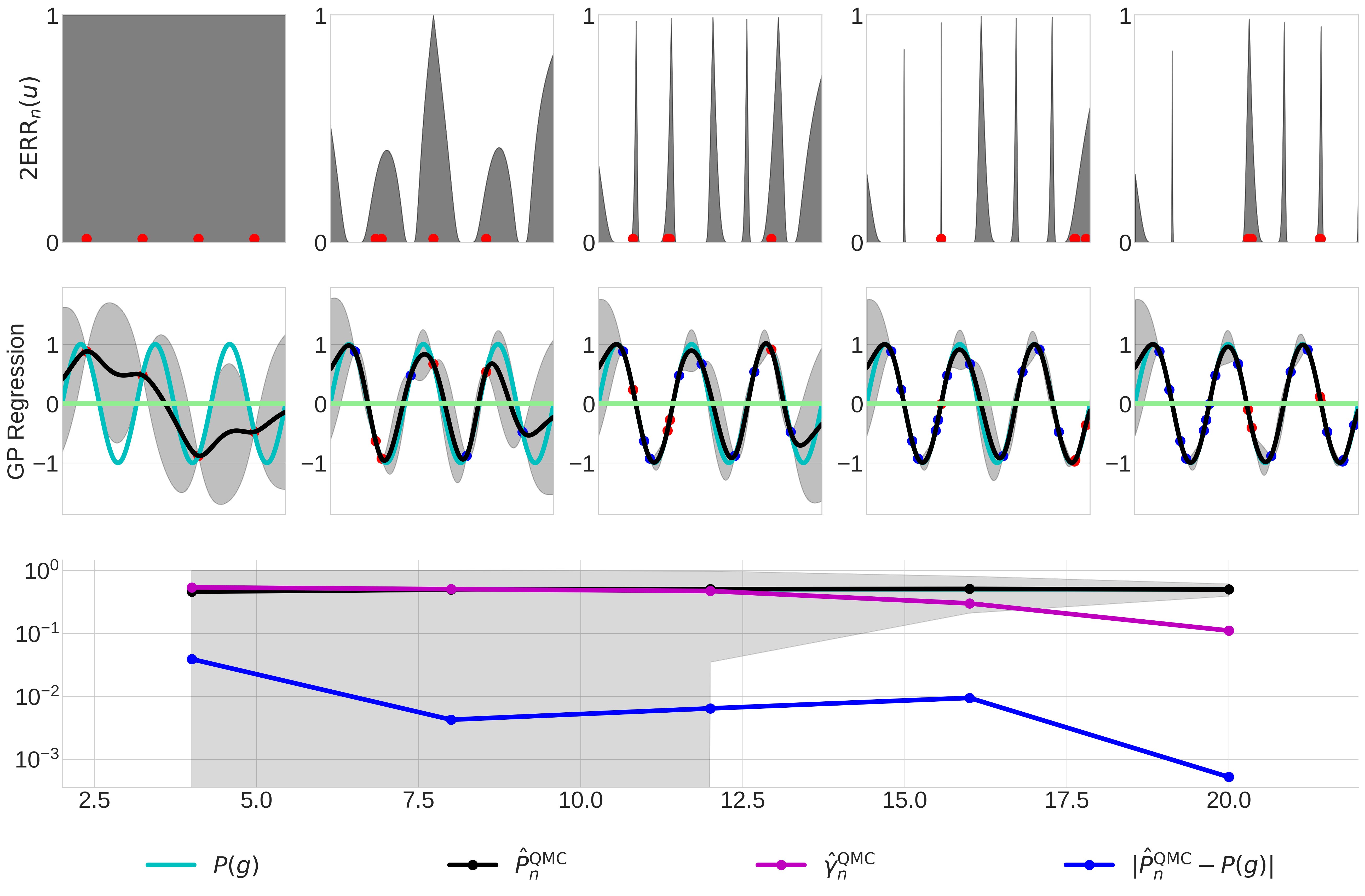}
    \caption{Moving left to right across columns traverses algorithm iterations. In each column, the top row shows the unnormalized density $2\ERR_n$ from which samples are drawn. The second row shows the true function in cyan with the corresponding Gaussian process visualized by its posterior mean in black and a 95\% pointwise confidence interval in gray. Red points are new points introduced in this iteration while blue points are those from previous iterations. The green line is the failure threshold. The bottom plot shows convergence of the algorithm with the resulting approximate credible interval in gray.}
    \label{fig:sine}
\end{figure}

\begin{figure}[H]
    \centering
    \textbf{Multimodal Problem}
    \includegraphics[width=\textwidth]{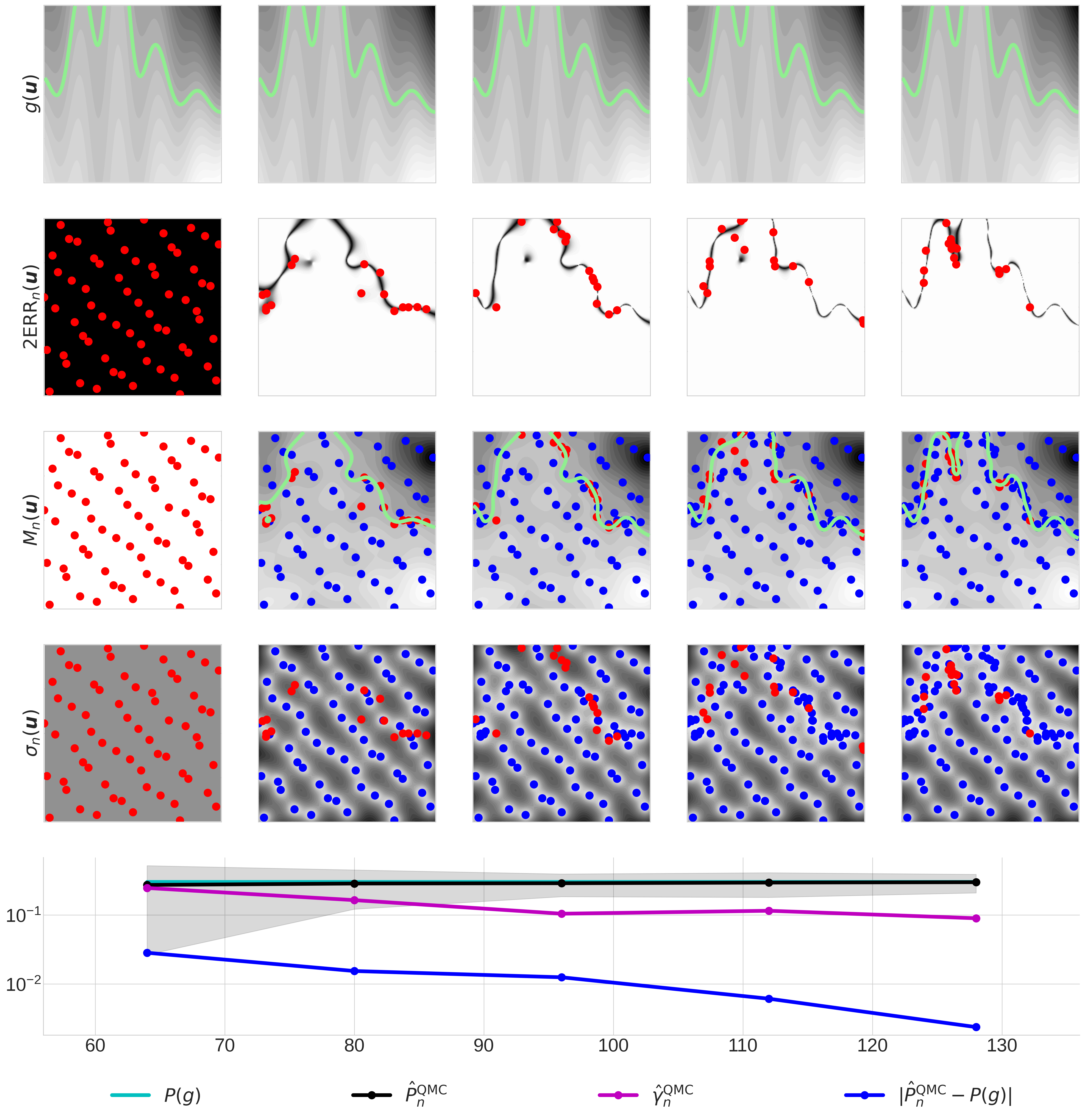}
    \caption{Moving left to right across columns traverses algorithm iterations. In each column, the top row visualizes the true function with true failure boundary $\{\bu \in [0,1]^d: g(\bu) = 0\}$ in green. The second row plots the unnormalized sampling density $2\ERR_n$ while the third and fourth row plot the Gaussian process posterior mean and standard deviation respectively. In the third row we plot the predicted failure boundary based on the posterior mean in green. The final row visualizes convergence as in \cref{fig:sine}.} 
    \label{fig:multimodal}
\end{figure}

\begin{figure}[htbp]
    \centering
    \textbf{Four Branch Problem}
    \includegraphics[width=\textwidth]{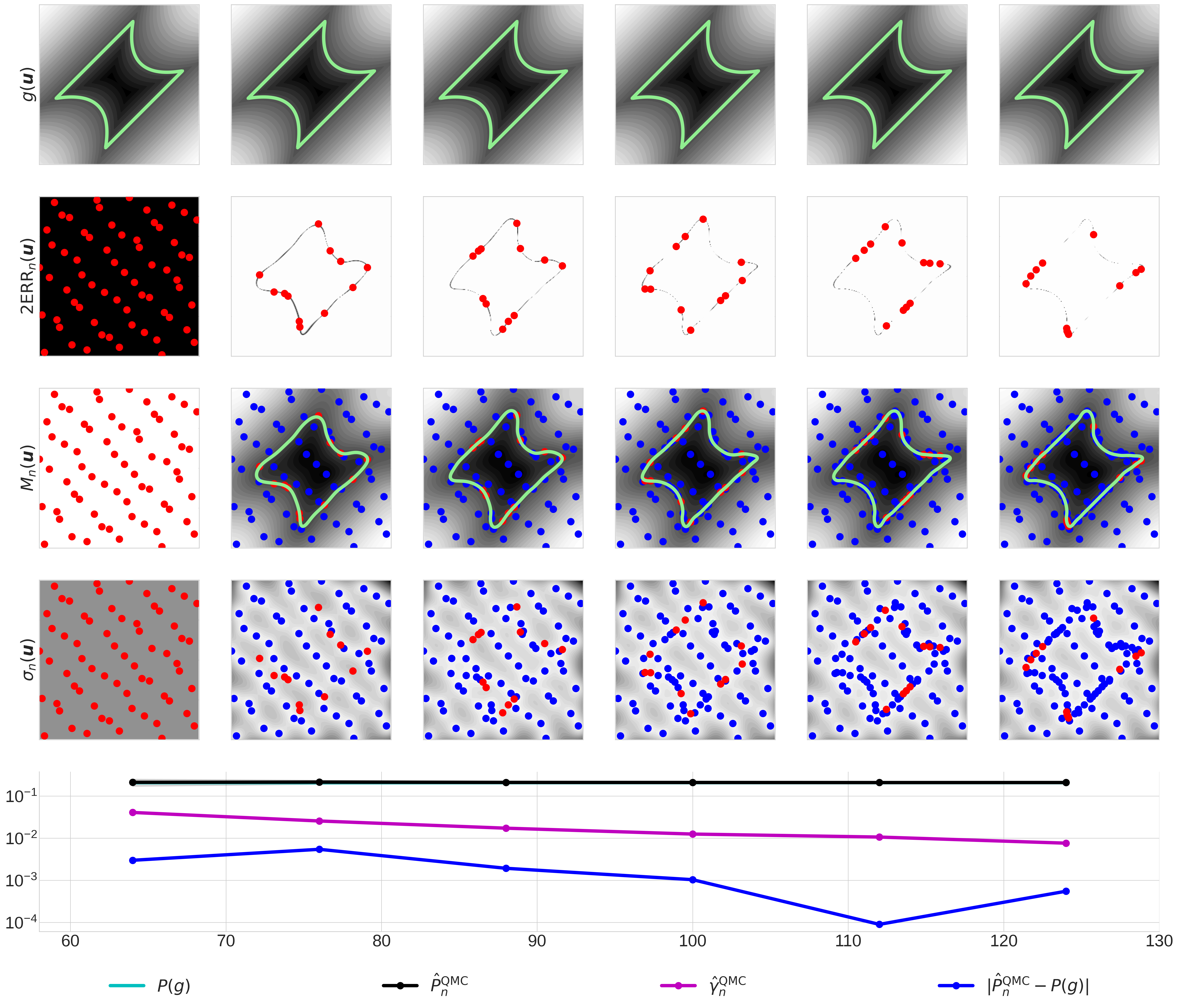}
    \caption{See \cref{fig:multimodal}.}
    \label{fig:fourbranch}
  \end{figure}

\begin{figure}[H]
    \centering
    \textbf{Ishigami Problem}
    \includegraphics[width=.8\textwidth]{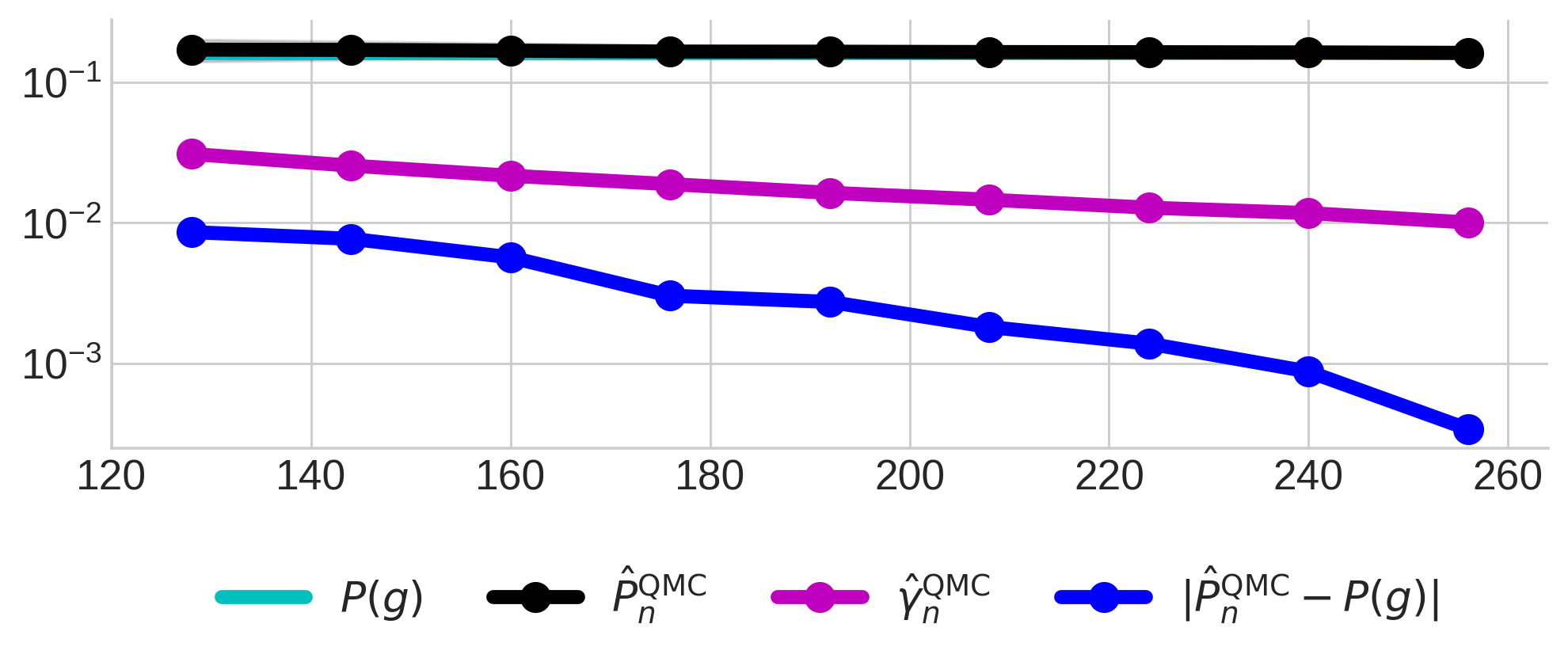}
    \caption{Convergence plots as in \cref{fig:sine}.}
    \label{fig:ishigami}
\end{figure}

\begin{figure}[H]
    \centering
    \textbf{Hartmann Problem}
    \includegraphics[width=.8\textwidth]{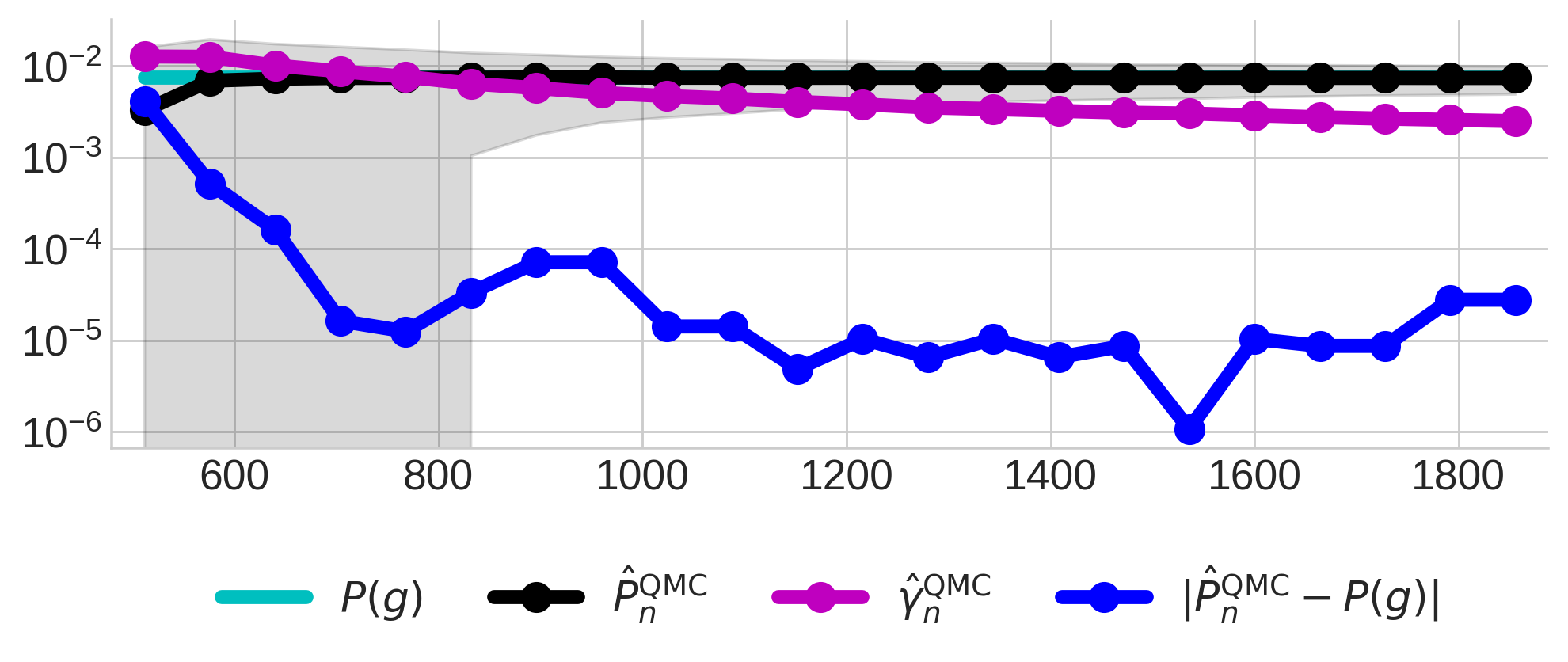}
    \caption{Convergence plots as in \cref{fig:sine}.}
    \label{fig:hartmann}
\end{figure}

\begin{figure}[H]
    \centering
    \textbf{Tsunami Problem}
    \includegraphics[width=\textwidth]{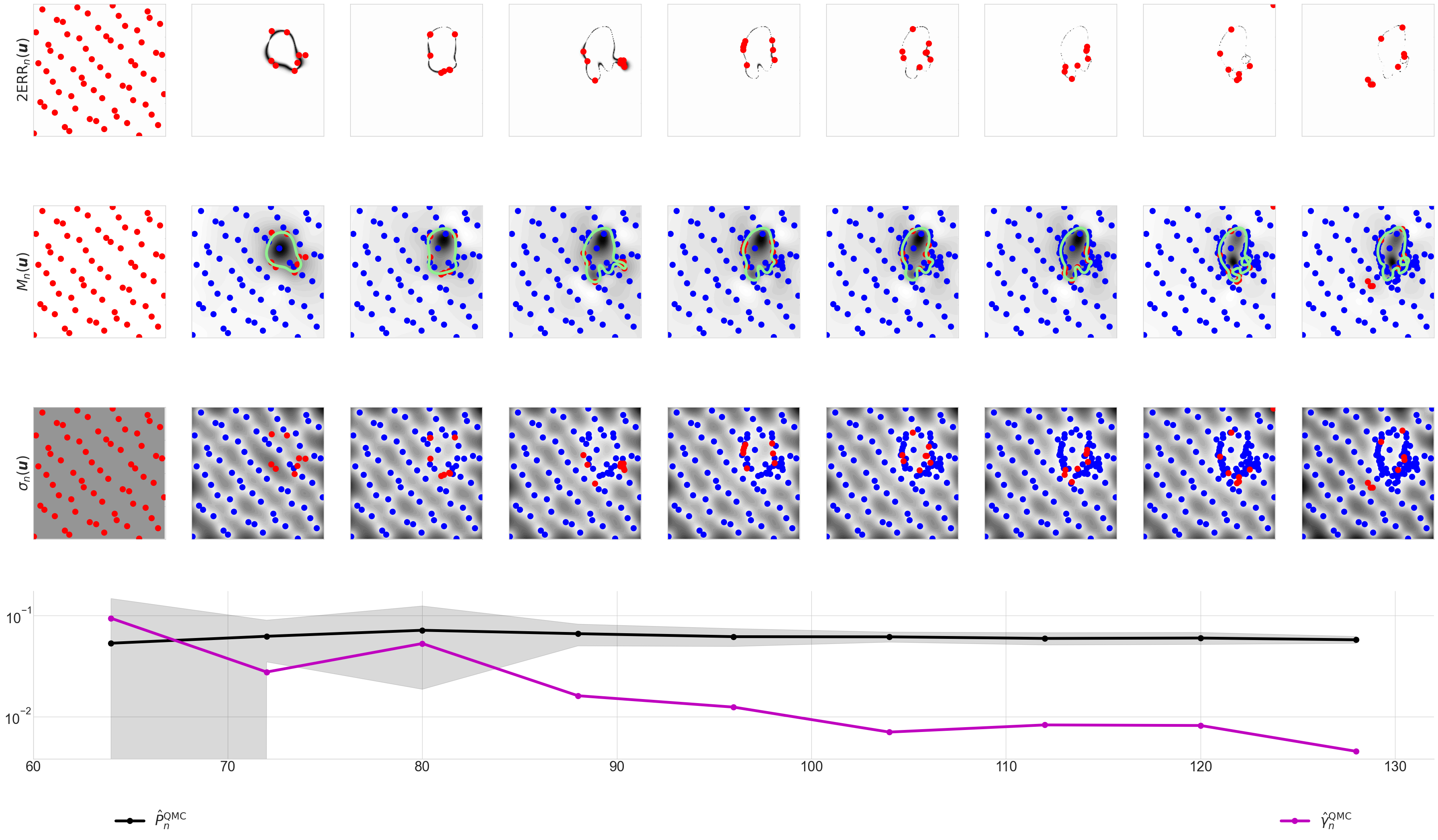}
    \caption{Similar to \cref{fig:multimodal}, except the first row is removed as we can no longer plot the true function. Also, since we no longer know the true mean the convergence plot in the bottom row cannot visualize our true error.}
    \label{fig:tsunami}
\end{figure}

\section{Conclusions and Future Work} \label{sec:conclusions_future_work}

In this article we derive credible intervals on the probability of failure based on a Gaussian process surrogate model. Our credible intervals are compatible with the variety of sampling schemes derived in the literature to intelligently refine the Gaussian process. We propose a novel sampling scheme that suggests arbitrarily sized batches of samples for the subsequent iteration, thus enabling functions capable of parallel evaluation on HPC systems to be queried at multiple locations simultaneously. We provide an open-source, scalable implementation capable of utilizing GPU acceleration for Gaussian process fitting. The algorithm shows strong empirical performance on both benchmark examples in low to medium dimensions and a realistic Tsunami simulation.

We propose a number of directions for future work:
\begin{itemize}
    \item The credible intervals quantify uncertainty in $P$. However, the credible interval bounds $\underline{P}$ and $\overline{P}$ cannot be computed exactly, so we resort to high precision QMC approximation. These QMC approximates come with their own small uncertainty which should be accounted for in more rigorous credible intervals. 
    \item There are many sampling schemes to refine the posterior distribution and thus shrink the width of the credible interval. A valuable contribution would be analytically or empirically comparing the convergence rates of credible interval width across sampling schemes. 
    \item There is a delicate interplay between the time it takes to run a simulation, fit a Gaussian process regression model, and propose samples. For simulations that are expensive to evaluate, it may be worthwhile to spend more time performing hyperparameter optimization and looking for good next sampling locations. For cheap functions, one may wish to match the sampling scheme and kernel in order to rapidly update the Gaussian process surrogate. Significant savings may be had by optimizing the tradeoff between time spent simulating and time spent between simulations.
    \item This article has discussed the simple case where the simulation is assumed to be deterministic. However, many functions are not deterministic but contain uncertainty of their own. Moreover, simulations may have cheaper, lower-fidelity counterparts. Extending the framework to noisy, multi-fidelity, and/or multi-level problems would greatly widen the scope of applicability and potentially yield significant computational savings. 
\end{itemize}

\section*{Acknowledgments}

The authors would like to thank Tomasz R. Bielecki for helpful discussions and guidance. 

\bibliographystyle{siamplain}
\bibliography{main}

\appendix

\section{Optimal Linear Combination of MC Estimates} \label{sec:optimal_linear_comb_ISMC_estimates}

Let $\{\hat{\mu}_t\}_{t \geq 0}$ be dependent, unbiased, individual estimates of $\mu$. For $\{\beta_t\}_{t \geq 0}$ chosen such that $\sum_{t \geq 0} \beta_t = 1$, the combined estimate $\hat{\mu} = \sum_{t \geq 0} \beta_t \hat{\mu}_t$ is still unbiased for $\mu$. The Cauchy-Schwarz inequality implies that  
\begin{align*}
  \Var[\hat{\mu}] &= \Var\left[\sum_{t \geq 0} \beta_t \hat{\mu}_t\right] \\
  &= \sum_{t \geq 0} \Var[\beta_t \hat{\mu}_t] + 2 \sum_{s < t} \Cov[\beta_s\hat{\mu}_s,\beta_t\hat{\mu_t}] \\
  &\leq \sum_{t \geq 0} \beta_t^2 \Var[\hat{\mu}_t] + 2\sum_{s < t}  \left\lvert \Cov[\beta_s\hat{\mu}_s,\beta_t\hat{\mu}_t]\right\rvert \\
  &\leq \sum_{t \geq 0} \beta_t^2 \Var[\hat{\mu}_t] + 2\sum_{s < t} \lvert \beta_s \rvert  \sqrt{\Var[\hat{\mu}_s]} \lvert \beta_t \rvert \sqrt{\Var[\hat{\mu}_t]} \\
  &= \left(\sum_{t \geq 0} \lvert \beta_t \rvert \sqrt{\Var[\hat{\mu}_t]}\right)^2
\end{align*}
Under the constraint that $\sum_{t \geq 0} \beta_t = 1$, the above bound is clearly minimized by setting $\beta_{t} = 1$ when $t = \argmin_t{\Var\left[\hat{\mu}_t\right]}$ and $\beta_t = 0$ otherwise. So the optimal unbiased combined estimate is simply the individual estimate with the lowest variance and thus does not incorporate samples outside those used to compute the lowest variance estimate. 

\section{Efficient Gaussian Process Updates} \label{sec:efficient_gp_updates}

\subsection*{Stage 1: Initializing the Gaussian Process}\label{sec:gp_s1}

We are first interested in initializing a Gaussian process given data $\bD_n = \{\mX,\by\}$ where $\mX \in [0,1]^{n \times d}$ and $\by \in \bbR^n$ so that $y_i = g(\bX_i)$ for $i=1,\dots,n$. The following procedure computes the common quantities required for this initial Gaussian Process to make future predictions, see \cite[Algorithm 2.1]{rasmussen.gp_ml}. Below we assume the prior mean and prior covariance each cost $\calO(d)$ to evaluate. 

\begin{algorithmic}
  \STATE{$\mK_{\mX,\mX} \gets k(\mX,\mX)+s^2\mI$ \COMMENT{$\calO(dn^2)$}}
  \STATE{$\mL \gets \mathrm{Cholesky}(\mK_{\mX,\mX})$ \COMMENT{$\calO(n^3)$}}
  \STATE{$\boldm_{\bX} \gets m(\mX)$ \COMMENT{$\calO(dn)$}}
  \STATE{$\bDelta \gets \by-\boldm_{\mX}$ \COMMENT{$\calO(n)$}}
  \STATE{$\bkappa \gets \mL\backslash \bDelta$ \COMMENT{$\calO(n^2)$}}
  \STATE{$\bbeta \gets \mL^\top \backslash\bkappa$ \COMMENT{$\calO(n^2)$}}
\end{algorithmic}
The total cost of this initialization is 
$$\calO(n^3+dn^2).$$
We may also perform hyperparameter optimization in this procedure with unchanged complexity. The remainder of this section requires fixed hyperparameters $\eta$. 

\subsection*{Stage 2: Predicting from the Initial Gaussian Process}\label{sec:gp_s2}

The following procedure computes at $\mU \in [0,1]^{N \times d}$ the posterior mean $\hat{\boldm}$ and standard deviation $\hat{\bsigma}$ given $\bD_n$, again see \cite[Algorithm 2.1]{rasmussen.gp_ml}.
\begin{algorithmic}
  \STATE{$\boldm_{\mU} \gets m(\mU)$ \COMMENT{$\calO(dN)$}}
  \STATE{$\mK_{\mX,\mU} \gets k(\mX,\mU)$  \COMMENT{$\calO(dnN)$}}
  \STATE{$\hat{\boldm} \gets \boldm_{\mU} + \mK_{\mX,\mU}^\top \beta$ \COMMENT{$\calO(nN)$}}
  \STATE{$\mV \gets \mL \backslash \mK_{\mX,\mU}$ \COMMENT{$\calO(n^2N)$}}
  \STATE{$\bv \gets \Diag(\mV^\top \mV)$ \COMMENT{$\calO(nN)$}}
  \STATE{$\bSigma_{\mU} \gets \sigma(\mU)$ \COMMENT{$\calO(dN)$}}
  \STATE{$\hat{\bsigma} \gets \sqrt{\bSigma_{\mU} - \bv}$ \COMMENT{$\calO(N)$}}
\end{algorithmic}
The above procedure costs 
$$\calO((dn+n^2)N).$$

\subsection*{Stage 3: Efficiently Updating the Gaussian Process}

Suppose we add $\tilde{\mX} \in \bbR^{b \times d}$, $\tilde{\by} \in \bbR^{b}$ to $\bD_n$ to get $\bD_{n + b}$ and want to update the GP with fixed hyperparameters. If we were to use the procedure in Stage 1, potentially including hyperparameter optimization, then the cost of refitting would be
$$\calO((n+b)^3+d(n+b)^2)=\calO(n^3+n^2b+nb^2+b^3+dn^2+dnb+db^2)$$
The following procedure efficiently updates common quantities by exploiting block matrix computations. 
\begin{algorithmic}
  \STATE{$\mL^{-1} \gets \mL \backslash \mI$ \COMMENT{$\calO(n^2)$}}
  \STATE{$\mK_{\tilde{\mX},\tilde{\mX}} \gets k(\tilde{\mX},\tilde{\mX})+s^2\mI$ \COMMENT{$\calO(db^2)$}}
  \STATE{$\mK_{\mX,\tilde{\mX}} \gets k(\mX,\tilde{\mX})$ \COMMENT{$\calO(dnb)$}}
  \STATE{$\tilde{\mL} \gets \mathrm{Cholesky}\left(\mK_{\tilde{\mX},\tilde{\mX}} - \mK_{\mX,\tilde{\mX}}^\top  (\mL^\top  \backslash (\mL \backslash \mK_{\mX,\tilde{\mX}}))\right)$ \COMMENT{$\calO(b^3+nb^2+n^2b)$}}
  \STATE{$\boldm_{\tilde{\mX}} \gets m(\tilde{\mX})$ \COMMENT{$\calO(db)$}}
  \STATE{$\tilde{\bk} \gets \tilde{\mL} \backslash \left(\tilde{\by}-\boldm_{\tilde{\mX}}-\mK_{\mX,\tilde{\mX}}^\top \mL^{-\top} \bkappa \right)$ \COMMENT{$\calO(b^2+n^2+nb)$}}
  \STATE{$\tilde{\bbeta}_2 \gets \tilde{\mL} \backslash \tilde{\bkappa}$ \COMMENT{$\calO(b^2)$}}
  \STATE{$\tilde{\bbeta}_1 \gets \mL \backslash \left(\bkappa - \mL^{-1}\mK_{\mX,\tilde{\mX}} \tilde{\bbeta}_2\right)$ \COMMENT{$\calO(n^2+nb)$}}
\end{algorithmic}
The above efficient procedure costs
$$\calO(db^2+dnb+b^3+nb^2+n^2b)$$
which avoids the cost of re-evaluating $\mK_{\mX,\mX}$, $\mathrm{Cholesky}(\mK_{\mX,\mX})$, and $\boldm_{\mX}$. This procedure is based on the block matrix relations
\begin{align*} 
  \mathrm{Cholesky}\begin{pmatrix} \mK_{\mX,\mX} & \mK_{\mX,\tilde{\mX}} \\ \mK_{\mX,\tilde{\mX}}^\top  & \mK_{\tilde{\mX},\tilde{\mX}} \end{pmatrix} &= \begin{pmatrix} \mL & 0 \\ \mK_{\mX,\tilde{\mX}}^\top \mL^{-\top} & \tilde{\mL} \end{pmatrix}, \\
  \begin{pmatrix} \mL & 0 \\ \mK_{\mX,\tilde{\mX}}^\top \mL^{-\top} & \tilde{\mL} \end{pmatrix} \begin{pmatrix} \bkappa \\ \tilde{\bkappa} \end{pmatrix} &= \begin{pmatrix} \by-\boldm_{\mX} \\ \tilde{\by}-\boldm_{\tilde{\mX}} \end{pmatrix}, \quad\mathrm{and} \\
  \begin{pmatrix} \mL & \mL^{-1}\mK_{\mX,\tilde{\mX}} \\ 0 & \tilde{\mL} \end{pmatrix} \begin{pmatrix} \tilde{\bbeta}_1 \\ \tilde{\bbeta}_2 \end{pmatrix} &= \begin{pmatrix} \bkappa \\ \tilde{\bkappa} \end{pmatrix}.
\end{align*}

\subsection*{Stage 4: Efficiently Updating Predictions}

Now suppose we would like to update the posterior mean and variance predictions at the same $\mU \in [0,1]^{N \times d}$ from stage 2 to $\hat{\tilde{\boldm}}$ and $\hat{\tilde{\bsigma}}$ now given $\bD_{n + b} = \{(\mX,\by)\} \cup \{(\tilde{\mX},\tilde{\by})\}$ from stage 3. Recomputing predictions from scratch using the procedure in stage 2 would cost 
$$\calO((d(n+b)+(n+b)^2)N) = \calO((dn+db+n^2+nb+b^2)N)$$
The following procedure efficiently updates predictions by exploiting block matrix computations.
\begin{algorithmic}
  \STATE{$\mK_{\tilde{\mX},\mU} \gets k(\tilde{\mX},\mU)$ \COMMENT{$\calO(dbN)$}}
  \STATE{$\hat{\tilde{\boldm}} \gets \boldm_{\mU} + \mK_{\mX,\mU}^\top \tilde{\bbeta}_1+\mK_{\tilde{\mX},\mU}^\top \tilde{\bbeta}_2$ \COMMENT{$\calO(nN+bN)$}}
  \STATE{$\tilde{\mV} \gets \tilde{\mL} \backslash (\mK_{\tilde{\mX},\mU} - \mK_{\mX,\tilde{\mX}}^\top \mL^{-\top}\mV)$ \COMMENT{$\calO(b^2N+n^2N+nbN)$}}
  \STATE{$\tilde{\bv} \gets \bv+\Diag(\tilde{\mV}^\top \tilde{\mV})$ \COMMENT{$\calO(bN)$}}
  \STATE{$\hat{\tilde{\bsigma}} \gets \sqrt{\bSigma_{\mU} - \tilde{\bv}}$ \COMMENT{$\calO(N)$}}
\end{algorithmic}
The above efficient procedure costs
$$\calO((db+b^2+n^2+nb)N)$$
which avoids recomputing $\boldm_{\mU}$, $\mK_{\mX,\mU}$, and $\bSigma_{\mU}$. This procedure is based on the block matrix relations
\begin{align*}
  \boldm_{\mU}+\begin{pmatrix} \mK_{\mX,\mU}^\top  & \mK_{\tilde{\mX},\mU}^\top  \end{pmatrix} \begin{pmatrix} \tilde{\bbeta}_1 \\ \tilde{\bbeta}_2 \end{pmatrix} &= \boldm_{\mU}+ \mK_{\mX,\mU}^\top \tilde{\bbeta}_1+\mK_{\tilde{\mX},\mU}^\top \tilde{\bbeta}_2, \\
  \begin{pmatrix} \mL & 0 \\ \mK_{\mX,\tilde{\mX}}^\top \mL^{-\top} & \tilde{\mL} \end{pmatrix} \begin{pmatrix} \mV \\ \tilde{\mV} \end{pmatrix} &= \begin{pmatrix} \mK_{\mX,\mU} \\ \mK_{\tilde{\mX},\mU} \end{pmatrix}, \quad\mathrm{and} \\
  \Diag\left(\begin{pmatrix} \mV^\top  & \tilde{\mV}^\top  \end{pmatrix} \begin{pmatrix} \mV \\ \tilde{\mV}\end{pmatrix}\right) &= \Diag\left(\mV^\top \mV\right) + \Diag\left(\tilde{\mV}^\top \tilde{\mV}\right).
\end{align*}

\end{document}

%% file: shared.tex
\usepackage{lipsum}
\usepackage{amsfonts}
\usepackage{graphicx}
\usepackage{epstopdf}
\usepackage{algorithmic}
\ifpdf
  \DeclareGraphicsExtensions{.eps,.pdf,.png,.jpg}
\else
  \DeclareGraphicsExtensions{.eps}
\fi

\usepackage{enumitem}
\setlist[enumerate]{leftmargin=.5in}
\setlist[itemize]{leftmargin=.5in}

\newsiamremark{remark}{Remark}
\newsiamremark{hypothesis}{Hypothesis}
\crefname{hypothesis}{Hypothesis}{Hypotheses}
\newsiamthm{claim}{Claim}
\newsiamremark{fact}{Fact}
\crefname{fact}{Fact}{Facts}

\headers{Credible Intervals for Probability of Failure with Gaussian Processes}{A. Sorokin, V. Rao}

\title{Credible Intervals for Probability of Failure with Gaussian Processes\thanks{Submitted to the editors March 13, 2026.
\funding{AS is partially supported by DARPA The Right Space HR0011-25-9-0031.}}}

\author{
  Aleksei Sorokin \thanks{Department of Statistics, University of Chicago, Chicago, IL 
  (\email{sorokin@uchicago.edu}).}
  \and 
  Vishwas Rao \thanks{Mathematics and Computer Science Division, Argonne National Laboratory, Lemont, IL 
  (\email{vhebbur@anl.gov}).}
}

\usepackage{amsopn}

\newcommand{\D}{\mathrm{d}}

\newcommand{\ba}{\boldsymbol{a}}
\newcommand{\bb}{\boldsymbol{b}}

\newcommand{\bD}{\boldsymbol{D}}

\newcommand{\bk}{\boldsymbol{k}}

\newcommand{\boldm}{\boldsymbol{m}}

\newcommand{\bu}{\boldsymbol{u}}
\newcommand{\bU}{\boldsymbol{U}}
\newcommand{\bv}{\boldsymbol{v}}
\newcommand{\bV}{\boldsymbol{V}}

\newcommand{\bx}{\boldsymbol{x}}
\newcommand{\bX}{\boldsymbol{X}}
\newcommand{\by}{\boldsymbol{y}}

\newcommand{\bsigma}{\boldsymbol{\sigma}}
\newcommand{\bSigma}{\boldsymbol{\Sigma}}

\newcommand{\bbeta}{\boldsymbol{\beta}}

\newcommand{\bDelta}{\boldsymbol{\Delta}}
\newcommand{\bkappa}{\boldsymbol{\kappa}}

\newcommand{\mV}{\mathsf{V}}
\newcommand{\mL}{\mathsf{L}}
\newcommand{\mK}{\mathsf{K}}
\newcommand{\mX}{\mathsf{X}}
\newcommand{\mI}{\mathsf{I}}
\newcommand{\mSigma}{\mathsf{\Sigma}}
\newcommand{\mA}{\mathsf{A}}
\newcommand{\mU}{\mathsf{U}}

\newcommand{\calB}{\mathcal{B}}

\newcommand{\calF}{\mathcal{F}}

\newcommand{\calN}{\mathcal{N}}
\newcommand{\calO}{\mathcal{O}}

\newcommand{\calT}{\mathcal{T}}
\newcommand{\calU}{\mathcal{U}}

\newcommand{\bbE}{\mathbb{E}}
\newcommand{\bbG}{\mathbb{G}}

\newcommand{\bbN}{\mathbb{N}}

\newcommand{\bbR}{\mathbb{R}}

\newcommand{\bbU}{\mathbb{U}}

\newcommand{\simiid}{\overset{\mathrm{IID}}{\sim}}

\newcommand{\Var}{\mathrm{Var}}
\newcommand{\Cov}{\mathrm{Cov}}

\newcommand{\Diag}{\mathrm{Diag}}

\DeclareMathOperator*{\argmin}{argmin}

\newcommand{\TP}{\mathrm{TP}}
\newcommand{\FP}{\mathrm{FP}}
\newcommand{\TN}{\mathrm{TN}}
\newcommand{\FN}{\mathrm{FN}}

\newcommand{\CMC}{\mathrm{CMC}}
\newcommand{\QMC}{\mathrm{QMC}}
\newcommand{\ISMC}{\mathrm{ISMC}}

\newcommand{\ACC}{\mathrm{ACC}}
\newcommand{\ERR}{\mathrm{ERR}}


%% file: main.bib
@misc{QMCPy.software,
  title               = {{QMCPy}: A {Q}uasi-{M}onte {C}arlo {P}ython library},
  author              = {Sou-Cheng T. Choi and Fred J. Hickernell and R. Jagadeeswaran and Michael J. McCourt and Aleksei G. Sorokin},
  url                 = {https://github.com/QMCSoftware/QMCSoftware},
}

@article{hlawka1961funktionen,
  title               = {Funktionen von beschr{\"a}nkter variatiou in der theorie der gleichverteilung},
  author              = {Hlawka, Edmund},
  year                = {1961},
  journal             = {Annali di Matematica Pura ed Applicata},
  publisher           = {Springer},
  volume              = {54},
  number              = {1},
  pages               = {325--333},
}

@article{Cornell_1968,
  title               = {Engineering seismic risk analysis},
  author              = {Cornell, C Allin},
  year                = {1968},
  journal             = {Bulletin of the Seismological Society of America},
  publisher           = {The Seismological Society of America},
  volume              = {58},
  number              = {5},
  pages               = {1583--1606},
}

@inproceedings{ishigami1990importance,
  title               = {An importance quantification technique in uncertainty analysis for computer models},
  author              = {Ishigami, Tsutomu and Homma, Toshimitsu},
  year                = {1990},
  booktitle           = {[1990] Proceedings. First International Symposium on Uncertainty Modeling and Analysis},
  pages               = {398--403},
  organization        = {IEEE},
}

@article{ang1992optimal,
  title               = {Optimal importance-sampling density estimator},
  author              = {Ang, George L and Ang, Alfredo H-S and Tang, Wilson H},
  year                = {1992},
  journal             = {Journal of engineering mechanics},
  publisher           = {American Society of Civil Engineers},
  volume              = {118},
  number              = {6},
  pages               = {1146--1163},
}

@book{niederreiter1992random,
  title               = {Random number generation and quasi-{M}onte {C}arlo methods},
  author              = {Niederreiter, Harald},
  year                = {1992},
  publisher           = {SIAM},
}

@article{hickernell1998generalized,
  title               = {A generalized discrepancy and quadrature error bound},
  author              = {Hickernell, Fred},
  year                = {1998},
  journal             = {Mathematics of computation},
  volume              = {67},
  number              = {221},
  pages               = {299--322},
}

@article{Easterling_2000,
  title               = {Climate extremes: observations, modeling, and impacts},
  author              = {Easterling, David R and Meehl, Gerald A and Parmesan, Camille and Changnon, Stanley A and Karl, Thomas R and Mearns, Linda O},
  year                = {2000},
  journal             = {Science},
  publisher           = {American Association for the Advancement of Science},
  volume              = {289},
  number              = {5487},
  pages               = {2068--2074},
}

@article{Easterling_2000A,
  title               = {Observed variability and trends in extreme climate events: a brief review},
  author              = {Easterling, David R and Evans, JL and Groisman, P Ya and Karl, Thomas R and Kunkel, Kenneth E and Ambenje, P},
  year                = {2000},
  journal             = {Bulletin of the American Meteorological Society},
  volume              = {81},
  number              = {3},
  pages               = {417--425},
}

@article{Vrouwenvelder_2000,
  title               = {Stochastic modelling of extreme action events in structural engineering},
  author              = {Vrouwenvelder, T},
  year                = {2000},
  journal             = {Probabilistic Engineering Mechanics},
  publisher           = {Elsevier},
  volume              = {15},
  number              = {1},
  pages               = {109--117},
}

@inproceedings{rasmussen.gp_ml,
  title               = {{G}aussian processes in machine learning},
  author              = {Rasmussen, Carl Edward},
  year                = {2003},
  booktitle           = {Summer school on machine learning},
  pages               = {63--71},
  organization        = {Springer},
}

@book{Ross_2003,
  title               = {A climatology of 1980--2003 extreme weather and climate events},
  author              = {Ross, Tom and Lott, Neal},
  year                = {2003},
  publisher           = {US Department of Commerece, National Oceanic and Atmospheric Administration, National Environmental Satellite Data and Information Service, National Climatic Data Center},
}

@book{Bucklew_2013B,
  title               = {Introduction to Rare Event Simulation},
  author              = {Bucklew, James A.},
  year                = {2004},
  publisher           = {Springer},
  address             = {New York},
  series              = {Springer Series in Statistics},
  doi                 = {10.1007/978-1-4757-4078-3},
}

@article{Au_2007,
  title               = {Application of subset simulation methods to reliability benchmark problems},
  author              = {Au, Siu Kui and Ching, J and Beck, JL},
  year                = {2007},
  journal             = {Structural safety},
  publisher           = {Elsevier},
  volume              = {29},
  number              = {3},
  pages               = {183--193},
}

@article{bichon2008efficient,
  title               = {Efficient global reliability analysis for nonlinear implicit performance functions},
  author              = {Bichon, Barron J and Eldred, Michael S and Swiler, Laura Painton and Mahadevan, Sandaran and McFarland, John M},
  year                = {2008},
  journal             = {AIAA Journal},
  volume              = {46},
  number              = {10},
  pages               = {2459--2468},
}

@article{Dysthe_2008,
  title               = {Oceanic rogue waves},
  author              = {Dysthe, Kristian and Krogstad, Harald E and M{\"u}ller, Peter},
  year                = {2008},
  journal             = {Annu. Rev. Fluid Mech.},
  publisher           = {Annual Reviews},
  volume              = {40},
  pages               = {287--310},
}

@inproceedings{Lesieutre_2008,
  title               = {Power system extreme event detection: The vulnerability frontier},
  author              = {Lesieutre, Bernard C and Pinar, Ali and Roy, Sandip},
  year                = {2008},
  booktitle           = {Hawaii International Conference on System Sciences, Proceedings of the 41st Annual},
  pages               = {184--184},
  organization        = {IEEE},
}

@inproceedings{Atputharajah_2009,
  title               = {Power system blackouts -- literature review},
  author              = {Atputharajah, Arulampalam and Saha, Tapan Kumar},
  year                = {2009},
  booktitle           = {2009 International Conference on Industrial and Information Systems (ICIIS)},
  pages               = {460--465},
  organization        = {IEEE},
}

@book{l2009monte,
  title               = {{M}onte {C}arlo and Quasi-{M}onte {C}arlo Methods 2008},
  author              = {L'Ecuyer, Pierre and Owen, Art B},
  year                = {2009},
  publisher           = {Springer},
}

@article{vazquez2009sequential,
  title               = {A sequential {B}ayesian algorithm to estimate a probability of failure},
  author              = {Vazquez, Emmanuel and Bect, Julien},
  year                = {2009},
  journal             = {IFAC Proceedings Volumes},
  publisher           = {Elsevier},
  volume              = {42},
  number              = {10},
  pages               = {546--550},
}

@book{dick2010digital,
  title               = {Digital nets and sequences: discrepancy theory and quasi--{M}onte {C}arlo integration},
  author              = {Dick, Josef and Pillichshammer, Friedrich},
  year                = {2010},
  publisher           = {Cambridge University Press},
}

@inproceedings{swiler2010importance,
  title               = {Importance sampling: Promises and limitations},
  author              = {Swiler, Laura and West, Nicholas},
  year                = {2010},
  booktitle           = {51st AIAA/ASME/ASCE/AHS/ASC Structures, Structural Dynamics, and Materials Conference 18th AIAA/ASME/AHS Adaptive Structures Conference 12th},
  pages               = {2850},
}

@article{bect2012sequential,
  title               = {Sequential design of computer experiments for the estimation of a probability of failure},
  author              = {Bect, Julien and Ginsbourger, David and Li, Ling and Picheny, Victor and Vazquez, Emmanuel},
  year                = {2012},
  journal             = {Statistics and Computing},
  publisher           = {Springer},
  volume              = {22},
  number              = {3},
  pages               = {773--793},
}

@article{dubourg2013metamodel,
  title               = {Metamodel-based importance sampling for structural reliability analysis},
  author              = {Dubourg, Vincent and Sudret, Bruno and Deheeger, Franois},
  year                = {2013},
  journal             = {Probabilistic Engineering Mechanics},
  publisher           = {Elsevier},
  volume              = {33},
  pages               = {47--57},
}

@article{echard2013combined,
  title               = {A combined importance sampling and kriging reliability method for small failure probabilities with time-demanding numerical models},
  author              = {Echard, Benjamin and Gayton, Nicolas and Lemaire, Maurice and Relun, Nicolas},
  year                = {2013},
  journal             = {Reliability Engineering \& System Safety},
  publisher           = {Elsevier},
  volume              = {111},
  pages               = {232--240},
}

@book{mcbook,
  title               = {{M}onte {C}arlo theory, methods and examples},
  author              = {Art B. Owen},
  year                = {2013},
  url                 = {https://artowen.su.domains/mc/},
}

@article{dalbey2014gaussian,
  title               = {{G}aussian process adaptive importance sampling},
  author              = {Dalbey, Keith R and Swiler, Laura P},
  year                = {2014},
  journal             = {International Journal for Uncertainty Quantification},
  publisher           = {Begel House Inc.},
  volume              = {4},
  number              = {2},
}

@article{l2014random,
  title               = {Random number generation and quasi-{M}onte {C}arlo},
  author              = {L'Ecuyer, Pierre},
  year                = {2014},
  journal             = {Wiley StatsRef: Statistics Reference Online},
  publisher           = {John Wiley \& Sons, Ltd Chichester, UK},
  pages               = {1--12},
}

@article{lv2015new,
  title               = {A new learning function for kriging and its applications to solve reliability problems in engineering},
  author              = {Lv, Zhaoyan and Lu, Zhenzhou and Wang, Pan},
  year                = {2015},
  journal             = {Computers \& Mathematics with Applications},
  publisher           = {Elsevier},
  volume              = {70},
  number              = {5},
  pages               = {1182--1197},
}

@book{rubinstein2016simulation,
  title               = {Simulation and the {M}onte {C}arlo Method},
  author              = {Rubinstein, Reuven Y. and Kroese, Dirk P.},
  year                = {2016},
  publisher           = {Wiley},
  address             = {Hoboken, NJ},
  series              = {Wiley Series in Probability and Statistics},
  edition             = {3rd},
}

@article{schobi2017rare,
  title               = {Rare event estimation using polynomial-chaos kriging},
  author              = {Sch{\"o}bi, Roland and Sudret, Bruno and Marelli, Stefano},
  year                = {2017},
  journal             = {ASCE-ASME Journal of Risk and Uncertainty in Engineering Systems, Part A: Civil Engineering},
  publisher           = {American Society of Civil Engineers},
  volume              = {3},
  number              = {2},
  pages               = {D4016002},
}

@inproceedings{gardner2018gpytorch,
  title               = {{GPyTorch}: Blackbox Matrix-Matrix {G}aussian Process Inference with {GPU} Acceleration},
  author              = {Jacob R. Gardner and Geoff Pleiss and Kilian Q. Weinberger and David Bindel and Andrew Gordon Wilson},
  year                = {2018},
  booktitle           = {Advances in Neural Information Processing Systems 31: Annual Conference on Neural Information Processing Systems 2018, NeurIPS 2018, December 3-8, 2018, Montr{\'{e}}al, Canada},
  pages               = {7587--7597},
  url                 = {https://proceedings.neurips.cc/paper/2018/hash/27e8e17134dd7083b050476733207ea1-Abstract.html},
  bibsource           = {dblp computer science bibliography, https://dblp.org},
  editor              = {Samy Bengio and Hanna M. Wallach and Hugo Larochelle and Kristen Grauman and Nicol{\`{o}} Cesa{-}Bianchi and Roman Garnett},
}

@article{jagadeeswaran2019fast,
  title               = {Fast automatic {B}ayesian cubature using lattice sampling},
  author              = {Jagadeeswaran, R and Hickernell, Fred J},
  year                = {2019},
  journal             = {Statistics and Computing},
  publisher           = {Springer},
  volume              = {29},
  number              = {6},
  pages               = {1215--1229},
}

@book{rathinavel2019fast,
  title               = {Fast automatic {B}ayesian cubature using matching kernels and designs},
  author              = {Rathinavel, Jagadeeswaran},
  year                = {2019},
  publisher           = {Illinois Institute of Technology},
}

@article{wahal2019bimc,
  title               = {BIMC: The {B}ayesian Inverse {M}onte {C}arlo method for goal-oriented uncertainty quantification. Part I},
  author              = {Wahal, Siddhant and Biros, George},
  year                = {2019},
  journal             = {ArXiv preprint},
  volume              = {abs/1911.00619},
  url                 = {https://arxiv.org/abs/1911.00619},
}

@article{bae.pf_gp_uncertainty_reduction,
  title               = {Estimating effect of additional sample on uncertainty reduction in reliability analysis using {G}aussian process},
  author              = {Bae, Sangjune and Park, Chanyoung and Kim, Nam H},
  year                = {2020},
  journal             = {Journal of Mechanical Design},
  publisher           = {American Society of Mechanical Engineers},
  volume              = {142},
  number              = {11},
  pages               = {111706},
}

@inproceedings{balandat2020botorch,
  title               = {Bo{T}orch: {A} Framework for Efficient {M}onte-{C}arlo {B}ayesian Optimization},
  author              = {Maximilian Balandat and Brian Karrer and Daniel R. Jiang and Samuel Daulton and Benjamin Letham and Andrew Gordon Wilson and Eytan Bakshy},
  year                = {2020},
  booktitle           = {Advances in Neural Information Processing Systems 33: Annual Conference on Neural Information Processing Systems 2020, NeurIPS 2020, December 6-12, 2020, virtual},
  url                 = {https://proceedings.neurips.cc/paper/2020/hash/f5b1b89d98b7286673128a5fb112cb9a-Abstract.html},
  bibsource           = {dblp computer science bibliography, https://dblp.org},
  editor              = {Hugo Larochelle and Marc'Aurelio Ranzato and Raia Hadsell and Maria{-}Florina Balcan and Hsuan{-}Tien Lin},
}

@inproceedings{seelinger2021high,
  title               = {High performance uncertainty quantification with parallelized multilevel Markov chain {M}onte {C}arlo},
  author              = {Seelinger, Linus and Reinarz, Anne and Rannabauer, Leonhard and Bader, Michael and Bastian, Peter and Scheichl, Robert},
  year                = {2021},
  booktitle           = {Proceedings of the International Conference for High Performance Computing, Networking, Storage and Analysis},
  pages               = {1--15},
}

@article{uribe2021cross,
  title               = {Cross-entropy-based importance sampling with failure-informed dimension reduction for rare event simulation},
  author              = {Uribe, Felipe and Papaioannou, Iason and Marzouk, Youssef M and Straub, Daniel},
  year                = {2021},
  journal             = {SIAM/ASA Journal on Uncertainty Quantification},
  publisher           = {SIAM},
  volume              = {9},
  number              = {2},
  pages               = {818--847},
}

@article{camera,
  title               = {CAMERA: A Method for Cost-aware, Adaptive, Multifidelity, Efficient Reliability Analysis},
  author              = {Renganathan, S. Ashwin and Rao, Vishwas and Navon, Ionel M.},
  year                = {2022},
  journal             = {ArXiv preprint},
  volume              = {abs/2203.01436},
  url                 = {https://arxiv.org/abs/2203.01436},
}

@incollection{jagadeeswaran2022fast,
  title               = {Fast Automatic {B}ayesian Cubature Using {S}obol' Sampling},
  author              = {Jagadeeswaran, Rathinavel and Hickernell, Fred J},
  year                = {2022},
  booktitle           = {Advances in Modeling and Simulation: Festschrift for Pierre L'Ecuyer},
  publisher           = {Springer},
  pages               = {301--318},
}

@inproceedings{qmc_software,
  title               = {Quasi-{M}onte {C}arlo Software},
  author              = {Choi, Sou-Cheng T and Hickernell, Fred J and Jagadeeswaran, Rathinavel and McCourt, Michael J and Sorokin, Aleksei G},
  year                = {2022},
  booktitle           = {International Conference on {M}onte {C}arlo and Quasi-{M}onte {C}arlo Methods in Scientific Computing},
  pages               = {23--47},
  organization        = {Springer},
}

@article{wagner2022rare,
  title               = {Rare event estimation using stochastic spectral embedding},
  author              = {Wagner, P-R and Marelli, Stefano and Papaioannou, Iason and Straub, Daniel and Sudret, Bruno},
  year                = {2022},
  journal             = {Structural Safety},
  publisher           = {Elsevier},
  volume              = {96},
  pages               = {102179},
}

@article{cheng2023rare,
  title               = {Rare event estimation with sequential directional importance sampling},
  author              = {Cheng, Kai and Papaioannou, Iason and Lu, Zhenzhou and Zhang, Xiaobo and Wang, Yanping},
  year                = {2023},
  journal             = {Structural Safety},
  publisher           = {Elsevier},
  volume              = {100},
  pages               = {102291},
}

@book{practicalqmc,
  title               = {Practical quasi-{M}onte {C}arlo integration},
  author              = {Art B. Owen},
  year                = {2023},
  publisher           = {\url{https://artowen.su.domains/mc/practicalqmc.pdf}},
}

@article{Seelinger2023,
  title               = {{UM}-{B}ridge: Uncertainty quantification and modeling bridge},
  author              = {Linus Seelinger and Vivian Cheng-Seelinger and Andrew Davis and Matthew Parno and Anne Reinarz},
  year                = {2023},
  journal             = {Journal of Open Source Software},
  publisher           = {The Open Journal},
  volume              = {8},
  number              = {83},
  pages               = {4748},
  doi                 = {10.21105/joss.04748},
  url                 = {https://doi.org/10.21105/joss.04748},
}

@article{sorokin.2025.ld_randomizations_ho_nets_fast_kernel_mats,
  title               = {{QMCPy}: a {P}ython software for randomized low-discrepancy sequences, quasi-{M}onte {C}arlo, and fast kernel methods},
  author              = {Aleksei G. Sorokin},
  year                = {2025},
  journal             = {ArXiv preprint},
  volume              = {abs/2502.14256},
  url                 = {https://arxiv.org/abs/2502.14256},
}
